\documentclass[a4paper,11pt]{amsart}
\usepackage[active]{srcltx}

\usepackage{graphicx}
\usepackage{amsmath}
\usepackage{amssymb}
\usepackage{hyperref}

\newtheorem{thm}{Theorem}[section]
\newtheorem{lem}[thm]{Lemma}%
\newtheorem{prop}[thm]{Proposition}%
\newtheorem{cor}[thm]{Corollary}%
\theoremstyle{remark}
\newtheorem{remark}{Remark}[section] %

\theoremstyle{plain}

\numberwithin{equation}{section}

\def\CC{{\mathbb C}}

\def\RR{{\mathbb R}}

\def\ZZ{{\mathbb Z}}

\def\Ad{\operatorname{Ad}}
\def\Span{\operatorname{Span}}

\def\veca{{\text{\boldmath$a$}}}
\def\vecb{{\text{\boldmath$b$}}}

\def\vecc{{\text{\boldmath$c$}}}

\def\vece{{\text{\boldmath$e$}}}

\def\vecell{{\text{\boldmath$\ell$}}}
\def\vecm{{\text{\boldmath$m$}}}

\def\vecq{{\text{\boldmath$q$}}}

\def\vect{{\text{\boldmath$t$}}}
\def\vecu{{\text{\boldmath$u$}}}
\def\vecv{{\text{\boldmath$v$}}}

\def\vecw{{\text{\boldmath$w$}}}
\def\vecx{{\text{\boldmath$x$}}}

\def\vecy{{\text{\boldmath$y$}}}
\def\vecz{{\text{\boldmath$z$}}}

\def\vecbeta{{\text{\boldmath$\beta$}}}
\def\vecgamma{{\text{\boldmath$\gamma$}}}

\def\vecxi{{\text{\boldmath$\xi$}}}
\def\veckappa{{\text{\boldmath$\kappa$}}}

\def\vecnull{{\text{\boldmath$0$}}}

\def\scrA{{\mathcal A}}
\def\scrB{{\mathcal B}}

\def\scrD{{\mathcal D}}
\def\scrE{{\mathcal E}}
\def\scrF{{\mathcal F}}

\def\scrK{{\mathcal K}}
\def\scrL{{\mathcal L}}

\def\scrN{{\mathcal N}}
\def\scrO{{\mathcal O}}
\def\scrQ{{\mathcal Q}}
\def\scrP{{\mathcal P}}

\def\scrV{{\mathcal V}}
\def\scrW{{\mathcal W}}

\def\fC{{\mathfrak C}}
\def\fD{{\mathfrak D}}

\def\fQ{{\mathfrak Q}}

\def\fS{{\mathfrak S}}

\def\fU{{\mathfrak U}}

\def\fZ{{\mathfrak Z}}

\def\diam{\operatorname{diam}}

\def\Proj{\operatorname{Proj}}

\def\diag{\operatorname{diag}}

\def\intl{{\operatorname{int}}}

\def\C{\operatorname{C{}}}

\def\L{\operatorname{L{}}}

\def\GL{\operatorname{GL}}

\def\S{\operatorname{S{}}}

\def\SL{\operatorname{SL}}
\def\ASL{\operatorname{ASL}}

\def\SO{\operatorname{SO}}

\def\T{\operatorname{T{}}}
\def\tr{\operatorname{tr}}

\def\supp{\operatorname{supp}}

\def\vol{\operatorname{vol}}

\def\GamG{\Gamma\backslash G}

\def\SLSL{\SL(n,\ZZ)\backslash\SL(n,\RR)}

\def\trans{\,^\mathrm{t}\!}

\def\Onder#1#2#3#4#5{#1 \setbox0=\hbox{$#1$}\setbox1=\hbox{$#2$}
       \dimen0=.5\wd0 \dimen1=\dimen0 \dimen2=\dp0 \dimen3=\dimen2
       \advance\dimen0 by .5\wd1 \advance\dimen0 by -#4
       \advance\dimen1 by -.5\wd1 \advance\dimen1 by -#4
       \advance\dimen2 by -#3 \advance\dimen2 by \ht1
       \advance\dimen2 by 0.3ex \advance\dimen3 by #5
        \kern-\dimen0\raisebox{-\dimen2}[0ex][\dimen3]{\box1}
       \kern\dimen1}

\newcommand{\GaG}{\Gamma\backslash G}

\newcommand{\lasl}{\mathfrak{asl}}
\newcommand{\lsl}{\mathfrak{sl}}
\newcommand{\ig}{\mathfrak{g}}
\newcommand{\ih}{\mathfrak{h}}

\newcommand{\Q}{\mathbb{Q}}
\newcommand{\R}{\mathbb{R}}
\newcommand{\Z}{\mathbb{Z}}

\newcommand{\minmod}{\text{ mod }}

\newcommand{\col}{\: : \:}

\newcommand{\bn}{\mathbf{0}}

\newcommand{\ve}{\varepsilon}

\newcommand{\matr}[4]{\left( \begin{matrix} #1 & #2 \\ #3 & #4 \end{matrix} \right) }
\newcommand{\smatr}[4]{\bigr( \begin{smallmatrix} #1 & #2 \\ #3 & #4 \end{smallmatrix} \bigr) }

\title{Free path lengths in quasicrystals}
\author{Jens Marklof}
\author{Andreas Str\"ombergsson}
\address{School of Mathematics, University of Bristol,
Bristol BS8 1TW, U.K.\newline
\rule[0ex]{0ex}{0ex} \hspace{8pt}{\tt j.marklof@bristol.ac.uk}}
\address{Department of Mathematics, Box 480, Uppsala University,
SE-75106 Uppsala, Sweden\newline
\rule[0ex]{0ex}{0ex} \hspace{8pt}{\tt astrombe@math.uu.se}}
\date{\today}
\thanks{J.M.\ is supported by a Royal Society Wolfson Research Merit Award, a Leverhulme Trust Research Fellowship and ERC Advanced Grant HFAKT. A.S.\ is a Royal Swedish Academy of Sciences Research Fellow supported by
a grant from the Knut and Alice Wallenberg Foundation.}

\begin{document}

\begin{abstract}
Previous studies of kinetic transport in the Lorentz gas have been limited to cases where 
the scatterers are distributed at random (e.g.\ at the points of a spatial Poisson process)
or at the vertices of a Euclidean lattice. 
In the present paper we investigate quasicrystalline scatterer configurations, which are non-periodic, yet strongly correlated. A famous example is the vertex set of the Penrose tiling. Our main result proves the existence of a limit distribution of the free path length, which answers a question of Wennberg. The limit distribution is characterised by a certain random variable on the space of higher dimensional lattices, and is distinctly different from the exponential distribution observed for random scatterer configurations. The key ingredients in the proofs are equidistribution theorems on homogeneous spaces, which follow from Ratner's measure classification.  
\end{abstract}

\maketitle

\section{Introduction}

\subsection{The setting}

The Lorentz gas is defined as an ensemble of non-interacting point particles moving in an array of spherical scatterers placed at the elements of a given point set $\scrP\subset\RR^d$ ($d\geq2$, and we assume that the scatterers do not overlap). Each particle travels with constant velocity along straight lines  until it collides with a scatterer, and is then reflected elastically. We denote by $\vecq(t),\vecv(t)$ the position and velocity of a particle at time $t$. Since the reflection is elastic, speed is a constant of motion; we may assume without loss of generality that $\|\vecv(t)\|=1$. The ``phase space'' is then the unit tangent bundle $\T^1(\scrK_\rho)$
where $\scrK_\rho\subset\RR^d$ is the complement of the set $\scrB^d_\rho + \scrP$ (the ``billiard domain''), and  $\scrB^d_\rho$ denotes the open ball of radius $\rho$, centered at the origin. We parametrize $\T^1(\scrK_\rho)$ by $(\vecq,\vecv)\in\scrK_\rho\times\S_1^{d-1}$, where we use the convention that for $\vecq\in\partial\scrK_\rho$ the vector $\vecv$ points away from the scatterer (so that $\vecv$ describes the velocity {\em after} the collision). 
The Liouville measure on $\T^1(\scrK_\rho)$ is 
\begin{equation} \label{LIOUVILLEDEF}
	d\nu(\vecq,\vecv)=d\!\vol_{\RR^d}(\vecq)\, d\!\vol_{\S_1^{d-1}}(\vecv)
\end{equation}
where $\vol_{\RR^d}$ and $\vol_{\S_1^{d-1}}$ refer to the Lebesgue measures on $\RR^d$ %
and $\S_1^{d-1}$, respectively.

The first collision time corresponding to the initial condition $(\vecq,\vecv)\in\T^1(\scrK_\rho)$ is 
\begin{equation} \label{TAU1DEF0}
	\tau_1(\vecq,\vecv;\rho) = \inf\{ t>0 : \vecq+t\vecv \notin\scrK_\rho \}. 
\end{equation}
Since all particles are moving with unit speed, we may also refer to $\tau_1(\vecq,\vecv;\rho)$ as the free path length. The distribution of free path lengths in the limit of small scatterer density (Boltzmann-Grad limit) has been studied extensively when $\scrP$ is a fixed realisation of a random point process (such as a spatial Poisson process) \cite{Boldrighini83,Gallavotti69,Polya18,Spohn78} and when $\scrP$ is a Euclidean lattice \cite{Boca03,Boca07,Bourgain98,Caglioti03,Dahlqvist97,Dettmann12,Golse00,partI,Nandori12,Polya18}. 
In the Boltzmann-Grad limit, the Lorentz process in fact converges to a random flight process, see \cite{Gallavotti69,Spohn78,Boldrighini83} for the case of random $\scrP$ and \cite{CagliotiGolse,partII,partIII,partIV} for periodic $\scrP$.

\subsection{Cut and project}
In the present work, we consider the Lorentz gas for scatterer configurations $\scrP$ given by regular cut-and-project sets; cf.\ \cite{Kraemer12,Wennberg12}. Examples of such $\scrP$ include large classes of quasicrystals, for instance the vertex set of the classical Penrose tiling. Further examples include all locally finite periodic point sets such as graphene's honeycomb lattice \cite{Boca09,Boca10}. 

To give a precise definition of cut-and-project sets in $\RR^d$, denote by $\pi$ and $\pi_\intl$ the orthogonal projection of $\RR^n=\RR^d\times\RR^m$ onto the first $d$ and last $m$ coordinates, and refer to $\RR^d$ and $\RR^m$ as the {\em physical space} and {\em internal space}, respectively. 
Let $\scrL\subset\RR^n$ be a lattice of full rank.
Then the closure of the set $\pi_\intl(\scrL)$ is an abelian subgroup $\scrA$ of $\RR^m$.
We denote by $\scrA^\circ$ the connected subgroup of $\scrA$ containing $\bn$;
then $\scrA^\circ$ is a linear subspace of $\R^m$,
say of dimension $m_1$, and there exist $\vecb_1,\ldots,\vecb_{m_2}\in\scrL$
($m=m_1+m_2$) such that $\pi_\intl(\vecb_1),\ldots,\pi_\intl(\vecb_{m_2})$ are linearly independent in
$\R^m/\scrA^\circ$ and 
\begin{align}
\scrA=\scrA^\circ+\Z\pi_\intl(\vecb_1)+\ldots+\Z\pi_\intl(\vecb_{m_2}).
\end{align}
We denote by $\mu_\scrA$ the Haar measure of $\scrA$, normalized so that its restriction to $\scrA^\circ$
is the standard $m_1$-dimensional Lebesgue measure.
We also set $\scrV=\R^d\times\scrA^\circ$, and note that $\scrL\cap\scrV$ is a lattice of full rank in $\scrV$.

Given $\scrL$ and a bounded subset $\scrW\subset\scrA$ with non-empty interior, we define
\begin{equation}\label{CUTPROJDEF}
	\scrP(\scrW,\scrL) = \{ \pi(\vecy) : \vecy\in\scrL, \; \pi_\intl(\vecy)\in\scrW \} \subset \RR^d .
\end{equation}
We will call $\scrP=\scrP(\scrW,\scrL)$ a {\em cut-and-project set}, and $\scrW$ the {\em window}. If $\scrW$ has boundary of measure zero with respect to $\mu_\scrA$, we will say $\scrP(\scrW,\scrL)$ is {\em regular}. 
It follows from Weyl equidistribution (see \cite{Hof98}; also Prop.\ \ref{HOFWEYLEXPLPROP} below) 
that for any regular cut-and-project set $\scrP$ and 
any bounded $\scrD\subset\RR^d$ with boundary of measure zero with respect to Lebesgue measure, 
\begin{equation}\label{WEYLCOUNTING1}
\lim_{T\to\infty} \frac{\#\{ \vecb \in \scrL\col \pi(\vecb)\in\scrP\cap T \scrD\}}{T^d} 
= \frac{\vol(\scrD) \mu_\scrA(\scrW)}{\vol(\scrV/(\scrL\cap\scrV))}.
\end{equation}
A further condition often imposed in the quasicrystal literature is that $\pi|_\scrL$ is injective (i.e., the map $\scrL\to \pi(\scrL)$ is one-to-one); we will not require this here. To avoid coincidences in $\scrP$, we simply assume in the following that the window is appropriately chosen so that the map $\pi_\scrW: \{ \vecy\in\scrL : \pi_\intl(\vecy)\in\scrW \}\to \scrP$ is bijective.
Then \eqref{WEYLCOUNTING1} implies
\begin{equation}\label{density000}
\lim_{T\to\infty} \frac{ \#(\scrP\cap T \scrD)}{T^d} = \frac{\vol(\scrD) \mu_\scrA(\scrW)}{\vol(\scrV/(\scrL\cap\scrV))} .
\end{equation}

Under the above assumptions $\scrP(\scrW,\scrL)$ is a Delone set, i.e., uniformly discrete and relatively dense in $\RR^d$.

We may obviously extend the definition of cut-and-project sets $\scrP(\scrW,\widetilde\scrL)$ to affine lattices $\widetilde\scrL=\scrL+\vecx$, for any $\vecx\in\RR^n$; note that $\scrP(\scrW,\scrL+\vecx)=\scrP(\scrW-\pi_\intl(\vecx),\scrL)+\pi(\vecx)$.

\subsection{The distribution of free path lengths in the Boltzmann-Grad limit}

In order to study the distribution of the free path length for random initial data $(\vecq,\vecv)$ we need to specify a probability measure on $\T^1(\scrK_\rho)$. A canonical choice is of course any Borel probability measure which is absolutely continuous with respect to the Liouville measure $\nu$. Given $s>0$ and a Borel probability measure $\Lambda$ on $\T^1(\R^d)$,
we define the family of Borel probability measures $\Lambda^{(s)}$ on $\T^1(\R^d)$ by
\begin{equation}
\Lambda^{(s)}(E)=\Lambda\bigl(\bigl\{(s^{-1}\vecq,\vecv)\col(\vecq,\vecv)\in E\bigr\}\bigr).
\end{equation}

\begin{thm}\label{Thm0}
Given any regular cut-and-project set $\scrP$ there is a non-increasing continuous function $F_\scrP: [0,\infty]\to [0,1]$ with $F_\scrP(0)=1$, $F_\scrP(\infty)=0$, such that for any Borel probability measure $\Lambda$ on $\T^1(\R^d)$ which is absolutely continuous with respect to Lebesgue measure, and any $s_0>0$, $\xi>0$,
we have
\begin{equation}
\Lambda^{(s)}(\{ (\vecq,\vecv)\in \T^1(\scrK_\rho) : \rho^{d-1} \tau_1(\vecq,\vecv;\rho) \geq \xi \})
\to F_\scrP(\xi),
\end{equation}
as $\rho\to0$, uniformly over all $s\geq s_0$.
\end{thm}

We highlight the fact that the limit distribution is independent of $\Lambda$. 
Our techniques will allow us to prove limit theorems for more singular measures. A natural example is to fix a generic point $\vecq\not\in\scrP$ and take $\vecv$ random:

\begin{thm}\label{Thm1}
Given any regular cut-and-project set $\scrP$ there is a subset $\fS\subset\R^d$ of Lebesgue measure zero such that
for any $\vecq\in\R^d\setminus\fS$, any $\xi>0$ and any
Borel probability measure $\lambda$ on $\S_1^{d-1}$ which is absolutely continuous with respect to Lebesgue measure,
we have
\begin{equation}\label{Thm1eq}
\lim_{\rho\to 0} \lambda(\{ \vecv\in\S_1^{d-1} : \rho^{d-1} \tau_1(\vecq,\vecv;\rho) \geq \xi \})
= F_\scrP(\xi) ,
\end{equation}
with $F_\scrP(\xi)$ as in Theorem \ref{Thm0}.
\end{thm}

In fact our proof shows that the limit in \eqref{Thm1eq} exists for \textit{every} $\vecq\in\R^d$;
however for $\vecq\in\fS$ the limit in general depends on $\vecq$.

Another possibility is to specify the location $\vecq\in\scrP$ of a scatterer and consider the initial data $(\vecq+\rho\vecbeta(\vecv),\vecv)$ on (or near) the scatterer's boundary, where $\vecbeta:\S_1^{d-1}\to\RR^{d}$ is some fixed continuous function and $\vecv$ is again chosen at random on $\S_1^{d-1}$. To avoid pathologies, we assume that
$(\vecbeta(\vecv)+\R_{>0}\vecv)\cap\scrB_1^d=\emptyset$
for all $\vecv\in\S_1^{d-1}$. We fix a map $K:\S_1^{d-1}\to\SO(d)$ such that
$\vecv K(\vecv)=\vece_1=(1,0,\ldots,0)$ for all $\vecv\in\S_1^{d-1}$;
we assume that $K$ is smooth when restricted to $\S_1^{d-1}$ minus
one point, cf.~\cite[Footnote 3, p.\ 1968]{partI}. We denote by $\vecx_\perp$ the orthogonal projection of $\vecx\in\RR^n$ onto $\{\vecnull\}\times\RR^{n-1}$, which is identified with $\RR^{n-1}$.

\begin{thm}\label{Thm2}
Given any regular cut-and-project set $\scrP$ and $\vecq\in\scrP$,  there is  a continuous function $F_{\scrP,\vecq}: [0,\infty]\times\RR_{\geq 0} \to [0,1]$ with $F_{\scrP,\vecq}(\,\cdot\,,r)$ non-increasing, $F_{\scrP,\vecq}(0,r)=1$, $F_{\scrP,\vecq}(\infty,r)=0$ for all $r\in\RR_{\geq 0}$,  such that
for any $\xi>0$ and any
Borel probability measure $\lambda$ on $\S_1^{d-1}$ which is absolutely continuous with respect to Lebesgue measure,
we have
\begin{equation}\label{Thm2eq}
\lim_{\rho\to 0} \lambda(\{ \vecv\in\S_1^{d-1} : \rho^{d-1} \tau_1(\vecq+\rho \vecbeta(\vecv),\vecv;\rho) \geq \xi \}) 
= \int_{\S_1^{d-1}}  F_{\scrP,\vecq}(\xi , \|(\vecbeta(\vecv)K(\vecv))_\perp\|)\,d\lambda(\vecv) .
\end{equation}
The convergence in \eqref{Thm2eq} is uniform over all $\vecq\in\scrP$.
\end{thm}

We remark that the proof actually shows that \eqref{Thm2eq} holds for any fixed $\vecq\in\pi(\scrL)$,
and uniformly over all $\vecq$ in any set of the form
$\pi(\scrL\cap\pi_\intl^{-1}(B))$ with $B$ a bounded subset of $\scrA$.

\subsection{Spaces of quasicrystals}\label{secQuasi}

We will now characterise the limit distributions in Theorems \ref{Thm1} and \ref{Thm2} in terms of  a certain homogeneous space $(\Gamma\cap H_g)\backslash H_g$ equipped with a translation-invariant probability measure $\mu_g$. In analogy with the space of Euclidean lattices of covolume one, $\SLSL$,  we will call such a space a {\em space of quasicrystals}. 

Set $G=\ASL(n,\R)=\SL(n,\RR)\ltimes\RR^n$, $\Gamma=\ASL(n,\Z)$. The multiplication law in $G$ is defined by
\begin{equation}
(M,\vecxi)(M',\vecxi')=(MM',\vecxi M'+\vecxi').
\end{equation}

For $g\in G$ we define an embedding of $\ASL(d,\RR)$ in $G$ by
\begin{equation}
\varphi_g: \ASL(d,\RR) \to G,\quad (A,\vecx) \mapsto g \left( \begin{pmatrix} A  &  0 \\ 0 & 1_m \end{pmatrix},(\vecx,\vecnull) \right) g^{-1} .
\end{equation}
We also set $G^1=\SL(n,\R)$ and $\Gamma^1=\SL(n,\Z)$, and identify $G^1$ with a subgroup of $G$ in the standard way;
similarly we identify $\SL(d,\R)$ with a subgroup of $\ASL(d,\R)$.
It follows from Ratner's theorems \cite{Ratner91a}, \cite[Cor.\ B]{Ratner91b} that there exists a (unique) closed connected subgroup $H_g$ of $G$ such that $\Gamma\cap H_g$ is a lattice in $H_g$, $\varphi_g(\SL(d,\R))\subset H_g$,
and the closure of $\Gamma\backslash\Gamma\varphi_g(\SL(d,\RR))$ in $\GamG$ is given by $\Gamma\backslash\Gamma H_g$. 
Note that $\Gamma\backslash\Gamma H_g$ can be naturally identified with the homogeneous space $(\Gamma\cap H_g)\backslash H_g$. 
We denote the unique right-$H_g$ invariant probability measure on either of these spaces by $\mu_g=\mu_{H_g}$.
Similarly, there exists a unique closed connected subgroup $\widetilde H_g$ of $G$ such that $\Gamma\cap \widetilde H_g$ is a lattice in $\widetilde H_g$, $\varphi_g(\ASL(d,\R))\subset\widetilde H_g$,
and the closure of $\Gamma\backslash\Gamma\varphi_g(\ASL(d,\RR))$ in $\GamG$ is given by $\Gamma\backslash\Gamma\widetilde H_g$. 
Note that $\Gamma\backslash\Gamma\widetilde H_g$ can be naturally identified with the homogeneous space $(\Gamma\cap\widetilde H_g)\backslash\widetilde H_g$. 
We denote the unique right-$\widetilde H_g$ invariant probability measure on either of these spaces by $\mu_{\widetilde H_g}$.
Of course, $H_g\subset \widetilde H_g$, and $\widetilde H_g=\widetilde H_{g(1_n,\vecx)}$ for any $\vecx\in\R^d\times\{\bn\}$.

Note that if $g\in G^1$ then $H_g\subset G^1$; in fact in this case $H_g$ is the unique 
closed connected subgroup
of $G^1$ such that $\Gamma^1\cap H_g$ is a lattice in $H_g$, $\varphi_g(\SL(d,\R))\subset H_g$,
and the closure of $\Gamma^1\backslash\Gamma^1\varphi_g(\SL(d,\RR))$ in $\Gamma^1\backslash G^1$ is given by 
$\Gamma^1\backslash\Gamma^1 H_g$.

Given $g\in G$ and $\delta>0$ we set $\scrL=\delta^{1/n}(\Z^ng)$ and let $\scrA=\overline{\pi_\intl(\scrL)}$ as before.
Then $\overline{\pi_\intl(\delta^{1/n}(\Z^nhg))}\subset\scrA$ for all $h\in\widetilde H_g$ and
$\overline{\pi_\intl(\delta^{1/n}(\Z^nhg))}=\scrA$ for $\mu_{\widetilde H_g}$-almost all $h\in\widetilde H_g$
and also for $\mu_g$-almost all $h\in H_g$; cf.\ Prop.\ \ref{SCRAPROP} and Prop.\ \ref{SLEQASLFORGENERICXPROP} below.
We fix $\delta>0$ and a window $\scrW\subset\scrA$, and consider the map from $\Gamma\backslash\Gamma\widetilde H_g$ to the set of point sets in $\RR^d$,
\begin{equation}\label{GGHQMAP}
\Gamma\backslash\Gamma h \mapsto \scrP(\scrW,\delta^{1/n}(\ZZ^n h g)).
\end{equation}
We denote the image of this map by $\widetilde\fQ_g=\widetilde\fQ_g(\scrW,\delta)$, and define a probability measure on $\widetilde\fQ_g$ as the push-forward of $\mu_{\widetilde H_g}$ (for which we will use the same symbol). 
This defines a random point process in $\R^d$ which is invariant under the natural action of $\ASL(d,\R)$ on $\RR^d$.
Similarly we denote by $\fQ_g=\fQ_g(\scrW,\delta)$ the image of $\Gamma\backslash\Gamma H_g$ under the map \eqref{GGHQMAP}, and define a probability measure on $\fQ_g$ as the push-forward of $\mu_g$; this again defines a random point process in $\R^d$, invariant under the natural action of $\SL(d,\R)$ on $\RR^d$.

We let $\fZ_\xi$ be the cylinder in $\R^d$ defined by
\begin{equation}\label{zyl}
	\fZ_\xi =\big\{(x_1,\ldots,x_d) \in\RR^d : 0<x_1<\xi, \; x_2^2+\ldots+x_d^2<1 \big\} .
\end{equation}
The following theorem provides formulas for the limit distributions in Theorems \ref{Thm0}, \ref{Thm1} and \ref{Thm2} in terms of $\widetilde H_g$ and $H_g$.

\begin{thm}\label{THM4}
Let $\scrP=\scrP(\scrL,\scrW)$ be a regular cut-and-project set, and $\vecq\in\RR^d$. Choose $g\in G$ and $\delta>0$ so that $\scrL-(\vecq,\bn)=\delta^{1/n}(\ZZ^n g)$. 
Then the function $F_\scrP(\xi)$ in Theorems \ref{Thm0} and  \ref{Thm1} is given by
\begin{equation}\label{Thm3eq}
F_\scrP(\xi) = \mu_{\widetilde H_g}(\{ \scrP'\in\widetilde\fQ_g \col \fZ_\xi \cap \scrP'  =\emptyset \}) .
\end{equation}
In fact if $\vecq\in\RR^d\setminus\fS$ (as in Theorem \ref{Thm1}), then $H_g=\widetilde H_g$ and this group is independent of the choice of $\vecq$.
On the other hand if $\vecq\in\scrP$, then the function $F_{\scrP,\vecq}(\xi,r)$ in Theorem \ref{Thm2} is given by
\begin{equation}\label{Thm3eqA}
F_{\scrP,\vecq}(\xi , r) = \mu_g(\{ \scrP'\in\fQ_g \col (\fZ_\xi+r\vece_d)\cap \scrP'  =\emptyset \}) 
\end{equation}
with $\vece_d=(0,\ldots,0,1)$.
\end{thm}

\subsection{The Siegel integral formula for quasicrystals}\label{SIEGELANNSEC}

The Siegel integral formula is a fundamental identity in the geometry of numbers \cite{siegel45,siegel}. 
We will prove an analogue for the space of quasicrystals. Let $f\in\L^1(\RR^d)$. Define for every $\scrP\in\fQ_g$ the Siegel transform 
\begin{equation}
\widehat f(\scrP) = \sum_{\vecq\in\scrP\setminus\{\vecnull\}} f(\vecq) .
\end{equation}
If $\scrL=\delta^{1/n}(\Z^ng)$ is a lattice then we set
\begin{align}
\delta_{d,m}(\scrL):=\frac1{\vol(\scrV/(\scrL\cap\scrV))}.
\end{align}
More generally if $\scrL$ is an affine lattice then we set $\delta_{d,m}(\scrL):=\delta_{d,m}(\scrL-\scrL)$;
note that $\scrL-\scrL$ is the lattice in $\R^n$ of which $\scrL$ is a translate.
\begin{thm}\label{SIEGELRDTHM}
Let $\scrL=\delta^{1/n}(\Z^ng)$ and $\fQ_g=\fQ_g(\scrW,\delta)$ as above,
and assume that $\scrP=\scrP(\scrW,\scrL)$ is regular and 
the map $\pi_\scrW:\{ \vecy\in\scrL : \pi_\intl(\vecy)\in\scrW \}\to \scrP$ is bijective.
Then for any $f\in\L^1(\R^d)$ we have
\begin{equation}\label{SIF}
\int_{\fQ_g}  \widehat f(\scrP) \,d\mu_g(\scrP) = \delta_{d,m}(\scrL)\mu_{\scrA}(\scrW) \int_{\RR^d} f(\vecx)\, d\!\vol_{\RR^d}(\vecx) .
\end{equation}
\end{thm}

The continuity for $\xi<\infty$ of the limit distributions $F_\scrP$ and $F_{\scrP,\vecq}$ in
Theorems \ref{Thm0}, \ref{Thm1} and \ref{Thm2} is an immediate consequence of Theorem \ref{SIEGELRDTHM} 
and the formulas in Theorem \ref{THM4}; for
$F_\scrP$ one uses also the fact that each $\widetilde\fQ_g$ can be obtained as $\fQ_{g'}$ for an
appropriate $g'$; cf.\ Proposition \ref{SLEQASLFORGENERICXPROP} and Corollary \ref{SIEGELTILDECOR} below.
We give a proof of the continuity at $\xi=\infty$ in Remark \ref{XIINFTYLIMREM}.

\subsection{Plan of the paper}

In Section \ref{EXsec} we give several examples of standard constructions of quasicrystals and discuss the corresponding
Ratner subgroups $\widetilde H_g$ and $H_g$ appearing in Theorem~\ref{THM4}.
In Section \ref{BASICSEC} we give some fundamental facts regarding the %
cut-and-project construction.
The key ingredient in the proofs of our main results are equidistribution theorems on the homogeneous space
$\Gamma\backslash\Gamma H_g$; these are established in Section \ref{DYNAMICSsec}.
In Section \ref{secSiegel} we prove the Siegel integral formula, Theorem \ref{SIEGELRDTHM}, in a slightly more general form,
and in Section~ \ref{MAINPROOFSSEC}, building on the results from previous sections, 
we prove %
Theorems \ref{Thm0}, \ref{Thm1}, \ref{Thm2} and \ref{THM4}.

Finally in an appendix we outline how the same methods can also be applied to understand the fine-scale statistics of directions in a cut-and-project set.

\section{Examples}\label{EXsec}

\subsection{Quasicrystals with low-dimensional internal spaces}

The following result holds:
\begin{prop}\label{GENERICPROP}
Assume $d>m$. Let $g\in G^1$ be such that for the lattice $\scrL=\Z^ng$, the map $\pi_{|\scrL}$ is injective and
$\scrA=\overline{\pi_\intl(\scrL)}=\R^m$.
Then $H_g=G^1$.
\end{prop}
The proof will be presented elsewhere.
The two assumptions on $\scrL$ %
(injectivity of $\pi_{|\scrL}$ and
density of $\overline{\pi_\intl(\scrL)}$) are standard in the quasicrystal literature.
It is important to note that the assumption $d>m$ in Proposition \ref{GENERICPROP} cannot be removed entirely;
indeed the number field construction to which we turn next can be used to give counterexamples for any $d,m$ with $d\mid m$.

In this vein, let us note that for arbitrary $d,m$, if $H_g=G^1$ then $\widetilde H_g=G$:
\begin{lem}\label{EASYTILDEHGLEM}
Let $g\in G^1$ be such that $H_g=G^1$. Then $\widetilde H_g=G$.
\end{lem}
\begin{proof}
It suffices to prove that $\Gamma\varphi_g(\ASL(d,\R))$ is dense in $G$.
Let $h\in G^1$ be given.
Since $H_g=G^1$, there exist sequences $\{\gamma_k\}\subset\Gamma^1$ and $\{A_k\}\subset\SL(d,\R)$ such that
$\gamma_k\varphi_g(A_k)\to h$ as $k\to\infty$.
Now for any $\vecell\in\Z^n$ and $\vecw\in\R^d$ we have $(1_n,\vecell)\gamma_k\in\Gamma$, $(A_k,\vecw)\in\ASL(d,\R)$, and
\begin{align*}
(1_n,\vecell)\gamma_k\varphi_g((A_k,\vecw))=(\gamma_k\varphi_g(A_k),\vecell\gamma_k\varphi_g(A_k)+(\vecw,\bn)g^{-1})
\to (h,\vecell h+(\vecw,\bn)g^{-1})
\end{align*}
as $k\to\infty$.
Thus the closure of $\Gamma\varphi_g(\ASL(d,\R))$ contains the set
\begin{align*}
\bigl\{(h,\vecv)\col h\in G^1,\: \vecv\in\Z^nh+(\R^d\times\{\bn\})g^{-1}\bigr\}.
\end{align*}
However $\Z^nh+(\R^d\times\{\bn\})g^{-1}$ is dense in $\R^n$ for almost every $h\in G^1$.
Hence $\Gamma\varphi_g(\ASL(d,\R))$ is dense in $G$.
\end{proof}

\subsection{Quasicrystals from algebraic number fields}
\label{NUMFIELDsec}

Let $K$ be a totally real number field of degree $N\geq2$ over $\Q$,
let $\scrO_K$ be its subring of algebraic integers,
and let $\pi_1,\ldots,\pi_N$  be the distinct embeddings $K$ into $\R$.
We will always view $K$ as a subset of $\R$ via $\pi_1$; in other words we agree that $\pi_1$ is the identity map.
Fix $d\geq1$ and set $n=dN$.
By abuse of notation we write $\pi_j$ also for the coordinate-wise embedding of $K^d$ into $\R^d$,
and for the entry-wise embedding of $M_d(K)$ (the algebra of $d\times d$ matrices with entries in $K$) into $M_d(\R)$.
Let $\scrL$  %
be the lattice in $\R^n=(\R^d)^N$ given by
\begin{align}\label{NUMFIELDLATTICE}
\scrL=\scrL_K^d:   %
=\Bigl\{(\vecx,\pi_2(\vecx),\ldots,\pi_N(\vecx))\col \vecx\in\scrO_K^d\Bigr\}.
\end{align}
As usual we set $m=n-d=(N-1)d$, let $\pi$ and $\pi_\intl$ be the projections of
$\R^n=(\R^d)^N=\R^d\times\R^m$ onto the first $d$ and last $m$ coordinates.
It follows from \cite[Cor.\ 2 in Ch.\ IV-2]{Weil} that $\pi_\intl(\scrL)$ is dense in $\R^m$,
i.e.\ we have $\scrA=\R^m$ and $\scrV=\R^n$ in the present situation.
Hence the window $\scrW$ should be taken as a subset of $\R^m$, and we consider
the cut-and-project set $\scrP(\scrW,\scrL)\subset\R^d$.

\subsubsection{Determining $H_g$ and $\widetilde H_g$}
Choose $\delta>0$ and $g\in G^1$ such that
\begin{align}
\scrL=\delta^{1/n}\Z^ng.
\end{align}
In fact
\begin{align}\label{DETOFLK}
\delta=|D_K|^{d/2},
\end{align}
where $D_K$ is the discriminant of $K$; %
cf., e.g., \cite[Ch.\ V.2, Lemma 2]{lang}.
We now claim that 
\begin{align}\label{HGNUMBERFIELDCASE}
\widetilde H_g=g\ASL(d,\R)^Ng^{-1}\qquad\text{and}\qquad
H_g=g\SL(d,\R)^Ng^{-1},
\end{align}
where $\ASL(d,\R)^N$ is embedded as a subgroup of $G=\ASL(n,\R)$ through
\begin{align}\label{SLdRNIMB}
\ASL(d,\R)^N\ni((A_1,\vecv_1),\ldots,(A_N,\vecv_N))\mapsto\Bigl(\diag[A_1,\ldots,A_N],(\vecv_1,\ldots,\vecv_N)\Bigr)\in G,
\end{align}
where $\diag[A_1,\ldots,A_N]$ is the block matrix whose diagonal blocks are $A_1,\ldots,A_N$ in this order,
and all other blocks vanish.

In order to prove \eqref{HGNUMBERFIELDCASE}, we set
\begin{align}
\Gamma_K^1:=\SL(d,\scrO_K)
\qquad\text{and}\qquad
\Gamma_K:=\ASL(d,\scrO_K)=\Gamma_K^1\ltimes\scrO_K^d,
\end{align}
which we view as subgroups of $\SL(d,\R)^N$ and $\ASL(d,\R)^N$, respectively, in the standard way through
$\gamma\mapsto(\pi_1(\gamma),\ldots,\pi_N(\gamma))$.
Then $\Gamma_K^1$ is a lattice in $\SL(d,\R)^N$ \cite[Thm.\ 12.3]{BoHari62}
and thus $\Gamma_K$ is a lattice in $\ASL(d,\R)^N$.
Note that $\Gamma_K$ stabilizes $\scrL$, i.e.\ $\scrL\gamma=\scrL$ holds for each $\gamma\in\Gamma_K$;
hence $\Gamma_K\subset g^{-1}\Phi_\delta(\Gamma) g\cap\ASL(d,\R)^N$,
where $\Phi_\delta$ is the isomorphism $G\to G$ given by $(A,\vecv)\mapsto(A,\delta^{1/n}\vecv)$.
It follows that $g^{-1}\Phi_\delta(\Gamma) g\cap\ASL(d,\R)^N$ is a lattice in $\ASL(d,\R)^N$, and thus
$\Gamma\cap\widetilde H$ is a lattice in $\widetilde H$, where $\widetilde H:=\Phi_\delta^{-1}(g\ASL(d,\R)^Ng^{-1})
=g\ASL(d,\R)^Ng^{-1}$.
Similarly $\Gamma^1\cap H$ is a lattice in $H$, where $H:=g\SL(d,\R)^Ng^{-1}$.
Using also the fact that $\Gamma_K^1$ is an irreducible lattice in $\SL(d,\R)^N$ it
follows that $\Gamma_K^1\varphi_1(\SL(d,\R))$ is dense in $\SL(d,\R)^N$
(cf.\ \cite[Cor.\ 5.21(5)]{Raghunathan}).
Conjugating with $g$ this implies that $(\Gamma^1\cap H)\varphi_g(\SL(d,\R))$ is dense in $H$,
or equivalently, $\Gamma^1\backslash\Gamma^1\varphi_g(\SL(d,\R))$ is dense in $\Gamma^1\backslash\Gamma^1 H$.
Hence $H$ has all the properties required of $H_g$, i.e.\ $H_g=H$.
Using also the fact that $\pi_\intl(\scrL_K^d)$ is dense in $\R^{(N-1)d}$ it follows similarly that
$\Gamma\backslash\Gamma\varphi_g(\ASL(d,\R))$ is dense in $\Gamma\backslash\Gamma\widetilde H$ and so 
$\widetilde H_g=\widetilde H$ and we have proved \eqref{HGNUMBERFIELDCASE}.

\subsubsection{}
Let us note that these considerations carry over trivially to the more general lattice
$\scrL_K^d g^0$, where $g^0=(g_1^0,\ldots,g_N^0)$ is any fixed element in $\GL(d,\R)^N$.
Indeed, note that $\scrL_K^d g^0={\delta'}^{1/n}\Z^ng'$ where
$\delta'=|D_K|^{d/2}\det g^0$ and $g'=g(\det g^0)^{-1/n}g^0\in G^1$,
and using the fact that conjugation by $g^0$ preserves $\varphi_1(\SL(d,\R))$ and
$\SL(d,\R)^N$, since $g^0$ is block diagonal,
we immediately verify that $H_{g'}=g'\SL(d,\R)^N{g'}^{-1}=H_g$;
similarly $\widetilde H_{g'}=g'\ASL(d,\R)^N{g'}^{-1}=\widetilde H_g$.

\subsection{Taking unions of translates of cut-and-project sets}
\label{UNIONSTRANSLSEC}

Let $\scrL\subset\R^n$ be an arbitrary lattice of full rank, and set $\scrA=\overline{\pi_\intl(\scrL)}$ as before;
fix a finite number of window sets $\scrW_1,\ldots,\scrW_s\subset\scrA$, and fix any vectors 
$\vect_1,\ldots,\vect_s\in\R^d$.
Let us consider the union of the translated cut-and-project sets $\vect_j+\scrP(\scrW_j,\scrL)$:
\begin{align}
\scrP\bigl(\{\scrW_j\},\{\vect_j\},\scrL\bigr):=\bigcup_{j=1}^s\bigl(\vect_j+\scrP(\scrW_j,\scrL)\bigr).
\end{align}

We will now show that, by a simple construction, the set $\scrP\bigl(\{\scrW_j\},\{\vect_j\},\scrL\bigr)$ 
can be recovered as a cut-and-project set $\scrP(\scrW',\scrL')$ within our framework.
We start by fixing a finite number of vectors $\vecb_1,\ldots,\vecb_r\in\R^n$ so that
$\{\vect_1,\ldots,\vect_s\}\subset\pi(\scrL+\sum_{k=1}^r\Z\vecb_k)$.
Note that this can always be achieved by taking $r=s$ and taking each $\vecb_k$ so that
$\pi(\vecb_k)=\vect_k$; however in practice one can often make more convenient choices with $r$ smaller than $s$.
Set $m'=m+r$ and $n'=d+m'$; let $\pi'$ and $\pi'_\intl$ be the projections of $\R^{n'}=\R^d\times\R^{m'}$
onto the first $d$ and last $m'$ coordinates, and let
\begin{align}\label{TRANSLUNIONCONSTR}
\scrL':=(\scrL\times\{\bn\}) %
+\sum_{k=1}^r\Z(\vecb_k,\vece_k)\subset\R^{n'},
\end{align}
where we express vectors using the decomposition $\R^{n'}=\R^n\times\R^{r}$, and $\vece_1,\ldots,\vece_r$ are
the standard basis vectors in $\R^r$.
Note that $\scrL'$ is a lattice of full rank in $\R^{n'}$.
We will call $\scrL'$ as in \eqref{TRANSLUNIONCONSTR} an \textit{extension} of rank $r$ over $\scrL$
by the \textit{extension vectors} $\{\vecb_k\}$.
Next let $\scrA'$ be the closure of $\pi'_\intl(\scrL')$ in $\R^{m'}$; %
then
\begin{align}
\scrA'=(\scrA\times\{\bn\}) %
+\sum_{k=1}^r\Z(\pi_\intl(\vecb_k),\vece_k),
\end{align}
where we express vectors using $\R^{m'}=\R^m\times\R^r$.
It follows from the choice of $\vecb_1,\ldots,\vecb_r$ that for each $j\in\{1,\ldots,s\}$ there exist
$\vecv^{(j)}\in\scrL$ and $\alpha_1^{(j)},\ldots,\alpha_r^{(j)}\in\Z$ such that
$\vect_j=\pi(\vecv^{(j)}+\sum_{k=1}^r\alpha_k^{(j)}\vecb_k)$.
We set
\begin{align}
\vecm_j':=(\vecv^{(j)},\bn)+\sum_{k=1}^r\alpha_k^{(j)}(\vecb_k,\vece_k)\in\scrL'
\quad\text{and}\quad
\scrW_j':=(\scrW_j\times\{\bn\})+\pi_\intl'(\vecm_j')\subset\scrA'
\end{align}
for $j=1,\ldots,s$.

As an immediate consequence of our definitions we now have
\begin{align}
\vect_j+\scrP(\scrW_j,\scrL)=\scrP(\scrW'_j,\scrL'),
\end{align}
and thus also, with $\scrW':=\cup_{j=1}^s\scrW_j'$:
\begin{align}\label{UNIONOFTRANSLATESRECOVERED}
\scrP\bigl(\{\scrW_j\},\{\vect_j\},\scrL\bigr):=\bigcup_{j=1}^s\bigl(\vect_j+\scrP(\scrW_j,\scrL)\bigr)
=\bigcup_{j=1}^s\scrP(\scrW'_j,\scrL')=\scrP(\scrW',\scrL'),
\end{align}
as desired. %

\vspace{5pt}

As a particular example, note that the above construction also applies when $m=0$,
in which case we understand $\scrA=\R^0=\{\bn\}$ and 
with the only possible (non-empty) $\scrW$ being $\scrW=\{\bn\}$, we have
$\scrP(\scrW,\scrL)=\scrL$.
Hence \eqref{UNIONOFTRANSLATESRECOVERED} shows that any \textit{periodic} Delone set
(viz.\ a union of a finite number of translates of a fixed lattice $\scrL\subset\R^d$)
can be obtained as a cut-and-project set \eqref{CUTPROJDEF}.
For example the case of a honeycomb recently treated by Boca and Gologan \cite{Boca09}
and Boca \cite{Boca10} is contained in the present work.
(In fact for the honeycomb all $\vect_j$ can be expressed as a rational linear combination of the lattice vectors in $\scrL$
and hence we are in the particularly simple situation described in the next Section \ref{RATTRANSLSEC}.)

We next discuss the Ratner subgroups associated with $\scrL'$.
Take $\delta>0$ and $g\in G^1$ so that $\scrL=\delta^{1/n}\Z^ng$.
Let $B$ be the $r\times n$ matrix whose row vectors are $\vecb_1,\ldots,\vecb_r$.
Then $B=\delta^{1/n}\beta g$ for some (uniquely determined) $\beta\in M_{r\times n}(\R)$,
and we have 
\begin{align}
\scrL'=\delta^{1/n'}\Z^{n'}g',
\end{align}
where
\begin{align}\label{RATTRANSLSECGPFORMULA}
g':=\delta^{-1/n'}\matr{\delta^{1/n}g}0B{1_r}
=\delta^{-1/n'}\matr{1_n}0{\beta}{1_r}\matr{\delta^{1/n}g}00{1_r}
\in\SL(n',\R).
\end{align}

\subsubsection{Determining $H_{g'}$ and $\widetilde H_{g'}$ -- in the special case of rational translates}
\label{RATTRANSLSEC}

Let us define the homomorphism $\phi_\beta:G\to \ASL(n',\R)$ through
\begin{align}\label{PHIBETADEF}
\phi_\beta((h,\vecv)):=\matr{1_n}0{\beta}{1_r}\left(\matr h00{1_r},(\vecv,\bn)\right){\matr{1_n}0{\beta}{1_r}}^{-1}
=\left(\matr h0{\beta h-\beta}{1_r},(\vecv,\bn)\right),
\end{align}
and note that 
\begin{align}\label{VARPHIGPFORMULA}
\varphi_{g'}((A,\vecv))=(\phi_\beta\circ\varphi_{g})((A,\delta^{(1/n')-(1/n)}\vecv)),\qquad\forall (A,\vecv)\in\ASL(d,\R).
\end{align}
Now assume that each $\vecb_k$ is a rational linear combination of the lattice vectors in $\scrL$.
We then claim that
\begin{align}\label{HGPFORCONGRQUASICR}
\widetilde H_{g'}=\phi_\beta(\widetilde H_{g})\qquad\text{and}\qquad H_{g'}=\phi_\beta(H_{g}).
\end{align}
Indeed, note that $\varphi_{g'}(\SL(d,\R))\subset\Phi_\beta(H_g)$ by \eqref{VARPHIGPFORMULA};
also the assumption about $\vecb_1,\ldots,\vecb_r$
implies that there is some positive integer $N$ such that $\beta\in M_{r\times n}(N^{-1}\Z)$,
and now one checks that $\SL(n',\Z)\cap\phi_\beta(H_g)$ contains $\phi_\beta(\Gamma^1(N)\cap H_g)$,
where $\Gamma^1(N)$ is the congruence subgroup
\begin{align}
\Gamma^1(N):=\bigl\{\gamma\in\Gamma^1\col \gamma\equiv 1_n\minmod N\Z\bigr\}.
\end{align}
It is known that
$\Gamma^1(N)$ has finite index in $\Gamma^1$; hence $\Gamma^1(N)\cap H_g$ has finite index in $\Gamma^1\cap H_g$,
and $\SL(n',\Z)\cap\phi_\beta(H_g)$ is a lattice in $\phi_\beta(H_g)$.
Next note that by Ratner \cite[Cor.\ B]{Ratner91b} there is a closed connected subgroup $H$ of $G^1$ 
such that $\varphi_g(\SL(d,\R))\subset H$, 
$\Gamma^1(N)\cap{H}$ is a lattice in ${H}$
and the closure of $(\Gamma^1(N)\cap{H})\varphi_{g}(\SL(d,\R))$ in $G^1$ equals ${H}$.
Then ${H}$ has all the properties required of $H_g$ and hence $H_g={H}$;
thus $H_{g}$ equals the closure of $(\Gamma^1(N)\cap H_g)\varphi_{g}(\SL(d,\R))$.
This implies that $\phi_\beta(H_{g})$ equals the closure of $(\SL(n',\Z)\cap\phi_\beta(H_g))\varphi_{g'}(\SL(d,\R))$,
and we have thus proved $H_{g'}=\phi_\beta(H_{g})$.
By an entirely similar argument, using $\Gamma^1(N)\ltimes\Z^n$ in place of $\Gamma^1(N)$, we also obtain
$\widetilde H_{g'}=\phi_\beta(\widetilde H_{g})$. Now \eqref{HGPFORCONGRQUASICR} is proved.

\subsubsection{Determining $H_g$ -- in a special case of linearly independent translates}
\label{IRRATTRANSLSEC}

Let us return to the special case of a \textit{periodic} Delone set,
i.e.\ a union of a finite number of translates of a fixed lattice $\scrL\subset\R^d$ ($d=n$).
Let $\vecb_1',\ldots,\vecb_d'$ be any fixed integer basis for $\scrL$.
We now consider the situation when the shift vectors $\vecb_1,\ldots,\vecb_r$ are such that
\textit{$\vecb_1,\ldots,\vecb_r,\vecb_1',\ldots,\vecb_d'$ are linearly independent over $\Q$.}
We claim that in this case, writing $n'=d+r$ and letting $g'\in\SL(n',\R)$ 
(as well as $g,\delta,B,\beta,\phi_\beta$) be as in Section \ref{RATTRANSLSEC},
we have 
\begin{align}\label{IRRATTRANSLCLAIM}
H_{g'}=\left\{\matr h0u{1_r}\col h\in G^1=\SL(d,\R),\: u\in M_{r\times d}(\R)\right\}
\end{align}
and
\begin{align}\label{IRRATTRANSLCLAIM2}
\widetilde H_{g'}=\bigl\{(T,(\vecv,\bn))\col T\in H_{g'},\:\vecv\in\R^d\bigr\}.
\end{align}

Indeed, let us write $H$ for the set in the right hand side of \eqref{IRRATTRANSLCLAIM};
then by \eqref{PHIBETADEF} and \eqref{VARPHIGPFORMULA} we have $\varphi_{g'}(\SL(d,\R))\subset H$;
also note that $\SL(n',\Z)\cap H$ is a lattice in $H$.
Thus to prove \eqref{IRRATTRANSLCLAIM}
it suffices to prove that $(\SL(n',\Z)\cap H)\phi_\beta(\SL(d,\R))$ is dense in $H$, i.e. that the set of matrices
\begin{align*}
\matr\gamma 0{\alpha}{1_r}\matr h0{\beta h-\beta}{1_r}
=\matr {\gamma h}0{(\alpha+\beta)h-\beta}{1_r},
\end{align*}
where $h,\gamma,\alpha$ vary over $G^1$, $\Gamma^1$ and $M_{r\times d}(\Z)$, respectively, is dense in $H$.
Replacing here $h$ by $\gamma^{-1}h$ and $\alpha$ by $\alpha\gamma$ we see that it suffices to prove that if 
$C\subset M_{r\times d}(\R/\Z)$ is the closure of the image of the set $\{\beta\gamma\col\gamma\in\Gamma^1\}$
under the projection $M_{r\times d}(\R)\to M_{r\times d}(\R/\Z)$, then $C=M_{r\times d}(\R/\Z)$.

Note that our assumption about $\vecb_1,\ldots,\vecb_r$ implies that there does not exist any $\veckappa\in\Z^r\setminus\{\bn\}$
satisfying $\veckappa \beta\in\Z^d$.
Hence by Weyl equidistribution, %
the set $\{\trans(\beta\trans\veca)\col\veca\in\Z^d\cap[1,T]^d\}$ becomes asymptotically equidistributed 
in the torus $\R^r/\Z^r$ as $T\to\infty$.
By a standard sieving argument, the same also holds if $\Z^d$ is replaced by
$\widehat\Z^d=\{\veca\in\Z^d\col \gcd(\veca)=1\}$, the set of primitive integer points,
and in particular we conclude that $\{\trans(\beta\trans\veca)\col\veca\in\widehat\Z^d\}$ is dense in $\R^r/\Z^r$.
Using also the compactness of $C$ and the fact that for any $\veca\in\widehat\Z^d$ there is some $\gamma\in\Gamma^1$ 
whose first column equals $\trans\veca$, it follows that for every $\vecw\in\R^r/\Z^r$ there is some $u\in C$ whose 
first column equals $\trans\vecw$.
Now let $v\in M_{r\times d}(\R/\Z)$ and $\ve>0$ be given.
Then there is some $\vecw\in\R^r/\Z^r$ such that $\Z\vecw$ is dense in $\R^r/\Z^r$ and $\trans\vecw$ is $\ve$-near
the first column of $v$.
By what we have just proved there is some $u\in C$ whose first column equals $\trans\vecw$.
But note that $C$ is $\Gamma^1$-right invariant; in particular $u\smatr 1{\veca}0{1_{d-1}}\in C$ for every $\veca\in\Z^{d-1}$;
and by choosing $\veca$ appropriately we can make each column of $u\smatr 1{\veca}0{1_{d-1}}$ be $\ve$-near the corresponding
column of $v$. Letting $\ve\to0$ we conclude $v\in C$,
i.e.\ we have proved that $C=M_{r\times d}(\R/\Z)$.
This completes the proof of \eqref{IRRATTRANSLCLAIM}.

Finally \eqref{IRRATTRANSLCLAIM2} follows immediately from \eqref{IRRATTRANSLCLAIM},
since $\widetilde H_{g'}$ contains both $H_{g'}$ and $\varphi_{g'}((1_d,\vecv))$
for each $\vecv\in\R^d$,
and since the right hand side of \eqref{IRRATTRANSLCLAIM2} is indeed a closed connected subgroup of $G$
which intersects $\ASL(n',\Z)$ in a lattice.

\subsection{Passing to a sublattice}
\label{SUBLATTICE}

Let $\scrL$ be as before, and let $\scrL'$ be a sublattice of $\scrL$ of full rank.
Then $\scrL'$ has finite index $N:=[\scrL:\scrL']$ as a subgroup of $\scrL$, 
and if $\scrL=\delta^{1/n}\Z^ng$ and $\scrL'={\delta'}^{1/n}\Z^ng'$ for some $\delta,\delta'>0$ and $g,g'\in G^1$ 
then $\delta'=N\delta$ and there is a $T\in M_n(\Z)$ with $\det T=N$ such that $g'=N^{-1/n}Tg$.
It will be convenient to know the precise relation between the Ratner subgroups for $\scrL$ and $\scrL'$:
\begin{lem}\label{HGFORSUBLATTICELEM}
In the situation just described, $H_{g'}=TH_gT^{-1}$ and $\widetilde H_{g'}=T\widetilde H_gT^{-1}$.
\end{lem}
\begin{proof}
By Cramer's rule we have $NT^{-1}\in M_n(\Z)$,
and by a simple computation this is seen to imply $T\Gamma^1(N)T^{-1}\subset\Gamma^1$.
Hence $\Gamma^1\cap TH_gT^{-1}$ contains $T(\Gamma^1(N)\cap H_g)T^{-1}$,
and it follows that $\Gamma^1\cap TH_gT^{-1}$ is a lattice in $TH_gT^{-1}$,
since $\Gamma^1(N)\cap H_g$ is a lattice in $H_g$.
Recall also that, as we noted in the proof of \eqref{HGPFORCONGRQUASICR},
$H_g$ equals the closure of $(\Gamma^1(N)\cap H_g)\varphi_{g}(\SL(d,\R))$.
Conjugating with $T$, this implies that
$TH_gT^{-1}$ equals the closure of $(\Gamma^1\cap TH_gT^{-1})\varphi_{g'}(\SL(d,\R))$.
Hence $TH_gT^{-1}$ has all the properties required of $H_{g'}$; hence $H_{g'}=TH_gT^{-1}$.
The proof of $\widetilde H_{g'}=T\widetilde H_gT^{-1}$ is entirely similar,
using $\Gamma^1(N)\ltimes N\Z^n$ in place of $\Gamma^1(N)$.
\end{proof}

\subsection{The quasicrystal associated with a Penrose tiling}

Let us now discuss the specific example of a quasicrystal associated with a Penrose tiling.
It is well-known that such a quasicrystal 
can be expressed as a regular cut-and-project set; cf.\ 
\cite{bruijn} and \cite[Sec.\ 6.4]{mS95}.
Specifically, set $d=2$, $m=3$, thus $n=5$, and we let $g\in \SO(5)$ be the orthogonal matrix whose row vectors are 
\begin{align}\label{PENROSEBJDEF}
\vecv_j=\sqrt{\tfrac25}\Bigl(
\cos(j\tfrac{2\pi}5),\sin(j\tfrac{2\pi}5),\cos(j\tfrac{4\pi}5),\sin(j\tfrac{4\pi}5),2^{-\frac12}\Bigr)
\qquad\text{for }\: j=0,1,2,3,4,
\end{align}
in this order,
and set $\scrL=\Z^5g$.
In other words, $\scrL$ is the lattice spanned by $\vecv_0,\ldots,\vecv_4$.
Now $\scrA=\pi_\intl(\scrL)=\R^2\times 5^{-\frac12}\Z$.
Indeed, 
if we set $\vecv_j':=p(\vecv_j)\in\R^2$ where $p:\R^5\to\R^2$ is the projection 
$(x_1,\ldots,x_5)\mapsto(x_3,x_4)$ then
$\vecv_2'=-\vecv_0'-\tau \vecv_1'$ and $\vecv_3'=\tau\vecv_0'+\tau\vecv_1'$, where $\tau:=\frac12(1+\sqrt5)$;
hence since $\tau\notin\Q$ we see that $\Z\vecv_0'+\Z\vecv_1'+\Z\vecv_2'+\Z\vecv_3'$ is dense in $\R^2$,
and this implies that %
$\pi_\intl(\scrL)$ is dense in $\R^2\times 5^{-\frac12}\Z$, as desired.
Next fix the window set
\begin{align}
\scrW:=\scrA\cap\pi_{\intl}((\scrQ_5+\vecgamma)g),
\end{align}
where $\scrQ_5$ is the open cube $(-\frac12,\frac12)^5$
and where $\vecgamma=(\gamma_1,\ldots,\gamma_5)$ is a fixed vector in $\R^5$ satisfying
$\sum_{j=1}^5\gamma_j\equiv\frac12\mod1$ and which is \textit{regular} in the sense 
that the subspace $(\R^2\times\{\bn\})g^{-1}$ does not meet any 2-face, edge or vertex of the cube
$\overline{\scrQ_5}+\vecgamma+\vecm$ for any $\vecm\in\Z^5$;
note that this condition is fulfilled for Lebesgue-almost all 
$\vecgamma$ with $\sum_{j=1}^5\gamma_j\equiv\frac12\mod1$.
With these choices $\scrP(\scrW,\scrL)$ is a quasicrystal associated with a Penrose tiling.
This is clear from Senechal \cite[Sec.\ 6.4]{mS95},
by noticing that the orthogonal transformation $g^{-1}$ maps $\scrL$ to $\Z^5$,
and maps the physical and internal spaces $\R^2\times\{\bn\}$ and $\{\bn\}\times\R^3$ 
onto $\scrE$ and $\scrE^\perp$, respectively, where
\begin{align*}
\scrE=\Span_\R\bigl\{(1,\cos(\tfrac{2\pi}5),\cos(\tfrac{4\pi}5),\cos(\tfrac{6\pi}5),\cos(\tfrac{8\pi}5)),\:
(1,\sin(\tfrac{2\pi}5),\sin(\tfrac{4\pi}5),\sin(\tfrac{6\pi}5),\sin(\tfrac{8\pi}5))\bigr\},
\end{align*}
and also maps $\scrW$ onto $\scrA g^{-1}\cap\Pi^\perp(\scrQ_5+\vecgamma)$,
where $\Pi^\perp$ denotes orthogonal projection onto $\scrE^\perp$.

We next wish to determine $H_g$. 
We will do so by observing that $\scrL$ can be obtained 
as an extension (cf.\ Sec.\ \ref{UNIONSTRANSLSEC}) of a sublattice (cf.\ Sec.\ \ref{SUBLATTICE}) of a 
number field lattice $\scrL_K^2$ as in Sec.\ \ref{NUMFIELDsec}.
Our discussion here is influenced by Pleasants \cite{Pleasants03}.

Let $K$ be the quadratic number field $K=\Q(\sqrt5)$ and let $\scrO_K$ be its ring of integers;
thus $\scrO_K=\Z[\tau]=\Z+\Z\tau$.
We write $a\mapsto\overline a$ for the conjugation map of $K$.
Let $\scrL_K^2$ be as in \eqref{NUMFIELDLATTICE}; thus
\begin{align*}
\scrL_K^2=\bigl\{(\alpha,\beta,\overline\alpha,\overline\beta\bigr)\col\alpha,\beta\in\scrO_K\bigr\}\subset\R^4.
\end{align*}
Let $\widetilde\scrL_K^2$ be the sublattice %
\begin{align*}
\widetilde\scrL_K^2:=\Bigl\{\bigl(\alpha,\beta,\overline\alpha,\overline\beta\bigr)\in\scrL_K^2\col %
\tr_{K/\Q}(\alpha+\beta)\in 5\Z\Bigr\},
\end{align*}
and let $\scrL'\subset\R^5$ be the rank-one extension (cf.\ \eqref{TRANSLUNIONCONSTR}) of $\widetilde\scrL_K^2$ by the
extension vector $\vecb=(1,0,1,0)\in\R^4$, i.e.\
\begin{align}\label{PENROSELPRIM}
\scrL'=\Bigl\{\bigl(\alpha,\beta,\overline\alpha,\overline\beta,k\bigr)\col\alpha,\beta\in\scrO_K,\: k\in\Z,\:
k\equiv-2\tr_{K/\Q}(\alpha+\beta)\:(\text{mod }5)\Bigr\}.
\end{align}

We now claim that 
\begin{align}\label{PENROSENUMBFIELDREL}
\scrL=\scrL'g_0,
\end{align}
where 
\begin{align}
g_0=\sqrt{\tfrac25}\begin{pmatrix} 1&&&&
\\
\cos(\frac{2\pi}5) & \sin(\frac{2\pi}5) &&&
\\
&&1&&
\\
&&\cos(\frac{4\pi}5) & \sin(\frac{4\pi}5) &
\\
&&&&2^{-\frac12}
\end{pmatrix}
\in\GL(5,\R).
\end{align}
To prove this relation we start by noticing that 
\begin{align}\label{LPBASIS}
\scrL'=\Z(1,0,1,0,1)
&+\Z(0,1,0,1,1)
+\Z(\tau,0,\overline\tau,0,-2)
+\Z(0,\tau,0,\overline\tau,-2)
+\Z(0,0,0,0,5).
\end{align}
Let us identify $\R^5$ with $\CC\times\CC\times\R$ through $(x_1,\ldots,x_5)\mapsto(x_1+ix_2,x_3+ix_4,x_5)$.
With this identification, we get from \eqref{LPBASIS} that
$\scrL'g_0=\sum_{j=0}^4\Z\vecu_j$ where
\begin{align}\notag
\vecu_0=\sqrt{\tfrac25}(1,1,2^{-\frac12}),\quad
&\vecu_1=\sqrt{\tfrac25}(\xi_5,\xi_5^2,2^{-\frac12}),\quad
&&\vecu_2=\sqrt{\tfrac25}(\tau,\overline{\tau},-2\cdot2^{-\frac12}),
\\
&\vecu_3=\sqrt{\tfrac25}(\tau\xi_5,\overline{\tau}\xi_5^2,-2\cdot2^{-\frac12}),
&&\vecu_4=\sqrt{\tfrac25}(0,0,5\cdot2^{-\frac12}),
\end{align}
with $\xi_5:=e^{2\pi i/5}$.
On the other hand the vectors $\vecv_j$ in \eqref{PENROSEBJDEF} are now given by
$\vecv_j=\sqrt{2/5}(\xi_5^j,\xi_5^{2j},2^{-\frac12})$,
and we recall that $\scrL=\sum_{j=0}^4\Z\vecv_j$.
Using $\tau=-\xi_5^2-\xi_5^3$ and $\overline\tau=1-\tau=-\xi_5-\xi_5^4$ we verify that
\begin{align}
\vecu_0=\vecv_0;\quad
\vecu_1=\vecv_1;\quad
\vecu_2=-\vecv_2-\vecv_3;\quad
\vecu_3=-\vecv_3-\vecv_4;\quad
\vecu_4=\vecv_0+\vecv_1+\vecv_2+\vecv_3+\vecv_4.
\end{align}
From these relations it is clear that $\scrL'g_0\subset\scrL$, and also, by a quick inspection,
that $\scrL\subset\scrL'g_0$, i.e.\ we have completed the proof of \eqref{PENROSENUMBFIELDREL}.

\vspace{5pt}

Now fix $g_K\in\SL(4,\R)$ so that $\scrL_K^2=5^{1/4}\Z^4g_K$
(cf.\ \eqref{DETOFLK} and note that $D_K=5$).
Then $\widetilde\scrL_K^2=5^{1/4}\Z^4Tg_K$ for some $T\in M_4(\Z)$ with $\det T=5$
(cf.\ Sec.\ \ref{SUBLATTICE}); also 
$\scrL'=5^{2/5}\Z^5g'$ where $g'\in G^1$ is given by (cf.\ \eqref{RATTRANSLSECGPFORMULA})
\begin{align}
g'=5^{-2/5}\matr{1_4}0{\vecbeta}1\matr{5^{1/4}Tg_K}001. %
\end{align}
Using \eqref{HGNUMBERFIELDCASE}, Lemma \ref{HGFORSUBLATTICELEM} and \eqref{HGPFORCONGRQUASICR}
(using $5\vecb\in\widetilde\scrL_K^2$), we have
\begin{align*}
\widetilde H_{g'}=\phi_\vecbeta(Tg_K\ASL(2,\R)^2g_K^{-1}T^{-1})=g'\widetilde H{g'}^{-1}
\qquad\text{and}\qquad
H_{g'}=g'H{g'}^{-1},
\end{align*}
where
\begin{align}
\widetilde H:=\left\{\left(\begin{pmatrix} A_1 && \\ &A_2& \\ &&1\end{pmatrix},(\vecv_1,\vecv_2,0)\right)
\col((A_1,\vecv_1),(A_2,\vecv_2))\in\ASL(2,\R)^2\right\}
\end{align}
and $H$ is the corresponding embedding of $\SL(2,\R)^2$ in $\SL(5,\R)$.
Finally, using \eqref{PENROSENUMBFIELDREL}
(which implies $g=5^{2/5}\gamma g'g_0$ for some $\gamma\in\Gamma^1$),
and the fact that conjugation with $g_0$ preserves each of $\widetilde H$, $\varphi_1(\ASL(2,\R))$, $H$ and 
$\varphi_1(\SL(2,\R))$
(since $g_0$ is $2,2,1$-block diagonal), we conclude:
\begin{align}\label{PENROSEHG}
\widetilde H_g=g\widetilde H g^{-1}\qquad\text{and}\qquad H_g=gHg^{-1}.
\end{align}

\section{Some basic observations}
\label{BASICSEC}
In this section we prove some basic facts which we will need later about the cut-and-project construction
and the related Ratner subgroup $H_g\subset G$.

\begin{prop}\label{DELONEPROP}
For any affine lattice $\scrL\subset\R^n$ and any bounded subset $\scrW\subset\scrA$ with nonempty interior,
the cut-and-project set $\scrP(\scrW,\scrL)$ is a Delone set.
\end{prop}
Cf.\ Meyer, \cite[%
p.\ 48 (Thm.\ IV)]{meyer}. For completeness we give a simple proof using our setup.  %
\begin{proof}
Using $\scrP(\scrW,\scrL+\vecx)=\scrP(\scrW-\pi_\intl(\vecx),\scrL)+\pi(\vecx)$ we may assume from start 
that $\scrL$ is a lattice.
Set $r=1+\diam(\scrW)$, and 
take $\delta>0$ so that
$\|\pi(\vecx)\|\geq\delta$ for all $\vecx\in\scrL\cap\scrB^n_r$ satisfying $\pi(\vecx)\neq\bn$.
Now let $\pi(\vecy)$ and $\pi(\vecy')$ (with $\vecy,\vecy'\in\scrL$, $\pi_\intl(\vecy),\pi_\intl(\vecy')\in\scrW$)
be any two distinct points in $\scrP$.
Then $\|\pi_\intl(\vecy)-\pi_\intl(\vecy')\|\leq\diam(\scrW)$;
hence if $\|\pi(\vecy)-\pi(\vecy')\|<1$ then $\vecy-\vecy'\in\scrL\cap\scrB_r^n$ and therefore
$\|\pi(\vecy)-\pi(\vecy')\|=\|\pi(\vecy-\vecy')\|\geq\delta$.
Thus $\|\pi(\vecy)-\pi(\vecy')\|\geq\min(1,\delta)$ always.
Hence $\scrP$ is uniformly discrete.

Next since $\scrW$ has non-empty interior, there is some $\vecb\in\scrL$ and an open ball $B\subset\scrA^\circ$
such that $\pi_\intl(\vecb)+B\subset\scrW$.
Let $B',B''\subset\scrA^\circ$ be open balls satisfying $B'+B''\subset B$.
Since the torus $\scrV/(\scrL\cap\scrV)$ is compact, there is a finite set $\{\vecv_1,\ldots,\vecv_r\}\subset\scrV$
such that $\vecv_j+(\scrB_1^d\times B')+(\scrL\cap\scrV)$ for $j=1,\ldots,r$ together cover $\scrV/(\scrL\cap\scrV)$.
It follows from the definition of $\scrV$ that $\R^d\times\{\bn\}$ is dense in $\scrV/(\scrL\cap\scrV)$;
in particular we can take $R>0$ so large that $\scrB_R^d\times\{\bn\}$ meets each set
$-\vecv_j+(\scrB_1^d\times B'')+(\scrL\cap\scrV)$,
or in other words $\vecv_j\in(\scrB_{R+1}^d\times B'')+(\scrL\cap\scrV)$ for each $j=1,\ldots,r$.
Now for every $\vecw\in\scrV$ we can take $j\in\{1,\ldots,r\}$ such that
\begin{align*}
\vecw\in\vecv_j+(\scrB_1^d\times B')+(\scrL\cap\scrV)
\subset(\scrB_{R+1}^d\times B'')+(\scrB_1^d\times B')+(\scrL\cap\scrV)
\subset(\scrB_{R+2}^d\times B)+(\scrL\cap\scrV).
\end{align*}
In particular for every $\vecx\in\R^d$, applying the above statement with $\vecw=(-\vecx+\pi(\vecb),\bn)$ we conclude that
$-\vecb+(\vecx+\scrB_{R+2}^d)\times(\pi_\intl(\vecb)+ B)$ has nonempty intersection with $\scrL$, and thus
$\vecx+\scrB_{R+2}^d$ has nonempty intersection with $\scrP$.
Hence $\scrP$ is relatively dense.
\end{proof}

\begin{prop}\label{HOFWEYLEXPLPROP}
Let $\scrL\subset\R^n$ be an affine lattice.
Then for any bounded subset $\scrW\subset\scrA$ %
with $\mu_\scrA(\partial\scrW)=0$ and any 
bounded subset $\scrD\subset\R^d$ with $\vol(\partial\scrD)=0$, we have
\begin{align*}
\frac{\#\bigl(\scrL\cap((\vecx+T\scrD)\times\scrW)\bigr)}{T^d}\to\delta_{d,m}(\scrL)\vol(\scrD)\mu_\scrA(\scrW),
\qquad\text{as }\: T\to\infty,
\end{align*}
uniformly over all $\vecx\in\R^d$.
\end{prop}
Cf.\ Schlottmann \cite{Schlottmann} and Hof \cite{Hof98}.
We give a proof along the lines of \cite{Hof98}.
\begin{proof}
By a translation argument we may assume without loss of generality that $\scrL$ is a lattice, i.e.\ $\bn\in\scrL$.
Furthermore, since $\scrW$ is bounded it can only intersect finitely many components of $\scrA$,
and by a partition and translation argument we may reduce to the situation when $\scrW\subset\scrA^\circ$.
Also by partitioning $\scrW$ further if necessary we may assume that
$\scrL\cap(\{\bn\}\times(\scrW-\scrW))=\{\bn\}$.

Since $\scrD$ is bounded and $\vol(\partial\scrD)=0$, $\scrD$ is Jordan measurable.
Hence for any $\eta>0$ we can construct bounded sets $\scrD_\eta^\pm\subset\R^d$ with 
$\vol(\partial\scrD_\eta^\pm)=0$ which have the properties that
$\scrD$ contains the $\eta$-neighbourhood of $\scrD_\eta^-$ and
$\scrD_\eta^+$ contains the $\eta$-neighbourhood of $\scrD$,
and $\vol(\scrD_\eta^+\setminus\scrD_\eta^-)\to0$ as $\eta\to0$.

Using now $\scrW\subset\scrA^\circ$ and $\scrL\cap(\{\bn\}\times(\scrW-\scrW))=\{\bn\}$ we have
\begin{align*}
T^{-d}\#\bigl(\scrL\cap((\vecx+T\scrD)\times\scrW)\bigr)
=T^{-d}\#\Bigl\{\vect\in T\scrD\col (\vecx+\vect,\bn)\in-(\{\bn\}\times\scrW)+(\scrL\cap\scrV)\Bigr\},
\end{align*}
and for any $T\geq1$ and any $\eta>0$ so small that the set $\scrB_\eta^d\times\scrW$ is injectively embedded
in the torus $\scrV/(\scrL\cap\scrV)$, the last quantity is bounded from above by
\begin{align}\label{HOFWEYLEXPLPROPPF1}
T^{-d}\vol(\scrB_\eta^d)^{-1}\vol\Bigl(\Bigl\{\vect\in T\scrD_\eta^+\col
(\vecx+\vect,\bn)\in-(\scrB_\eta^d\times\scrW)+(\scrL\cap\scrV)\Bigr\}\Bigr),
\end{align}
and from below by the analogous expression with $\scrD_\eta^-$.
However the fact that $\pi_\intl(\scrL\cap\scrV)$ is dense in $\scrA^\circ$ implies that
$\R^d\times\{\bn\}$ is not contained in any subspace of $\scrV=\R^d\times\scrA^\circ$ 
spanned by $\scrL$-vectors, other than $\scrV$ itself.
Therefore, by Weyl equidistribution, the set $(\vecx+T\scrD_\eta^+)\times\{\bn\}$ becomes asymptotically equidistributed
in the torus $\scrV/(\scrL\cap\scrV)$ as $T\to\infty$;   %
and in particular 
(since both $\scrW$ and $\scrD_\eta^+$ are Jordan measurable)
the expression in \eqref{HOFWEYLEXPLPROPPF1} tends to
\begin{align*}
\frac{\vol(\scrD_\eta^+)\vol(\scrB_\eta^d\times\scrW+(\scrL\cap\scrV)/(\scrL\cap\scrV))}
{\vol(\scrB_\eta^d)\vol(\scrV/(\scrL\cap\scrV))}
=\delta_{d,m}(\scrL)\vol(\scrD_\eta^+)\mu_\scrA(\scrW),
\end{align*}
uniformly with respect to $\vecx\in\R^d$,
where the last equality holds since $\scrB_\eta^d\times\scrW$ is injectively embedded in $\scrV/(\scrL\cap\scrV)$.
Similarly our bound from below tends to $\delta_{d,m}(\scrL)\vol(\scrD_\eta^-)\mu_\scrA(\scrW)$.
The proof is completed by taking $\eta\to0$ and using %
$\vol(\scrD_\eta^\pm)\to\vol(\scrD)$.
\end{proof}

\begin{thm}\label{RATNERCOUNTABLETHM}
(Ratner \cite{Ratner91b}.) The family $\{H_g\col g\in G\}$ is countable.
\end{thm}
\begin{proof}
This follows from \cite[Cor.\ A(2)]{Ratner91b}
(for note that by \cite[Cor.\ B]{Ratner91b}, for each $g\in G$ there is a one-parameter subgroup
$U$ of $\varphi_{g}(\SL(d,\R))$ which is unipotent in $G$ and such that
$\overline{\Gamma\backslash\Gamma U}=\Gamma\backslash\Gamma H_{g}$).
\end{proof}

\begin{cor}\label{RATNERCOUNTABLETHMCOR}
Let $g\in G$. Then $H_{hg}\subset H_g$ for all $h\in H_g$ and $H_{hg}=H_g$ for $\mu_g$-almost all $h\in H_g$.
\end{cor}
\begin{proof}
For any $h\in H_g$ we have $H_{hg}\subset H_g$ since $\varphi_{hg}(\SL(d,\R))\subset hH_gh^{-1}=H_g$.
On the other hand, let $U$ be a one-parameter subgroup of $\varphi_{g}(\SL(d,\R))$ which is unipotent in $G$ and such that
$\overline{\Gamma\backslash\Gamma U}=\Gamma\backslash\Gamma H_{g}$
\cite[Cor.\ B]{Ratner91b}.
Then for any $h\in H_g$ we have $hUh^{-1}\subset\varphi_{hg}(\SL(d,\R))$, and so
$\Gamma hU\subset\Gamma\backslash\Gamma H_{hg}h$. Therefore if $H_{hg}\subsetneq H_g$ then
$\Gamma hU$ is not dense in $\Gamma\backslash\Gamma H_{g}$.
Since $U$ acts ergodically on $(\Gamma\backslash\Gamma H_g,\mu_g)$
\cite[Cor.\ A]{Ratner91b} this can only happen for a $\mu_g$-null set of $h\in H_g$.
\end{proof}

\begin{prop}\label{SCRAPROP}
Let $g\in G$ and set $\scrL=\Z^ng$ and $\scrA=\overline{\pi_\intl(\scrL)}$. %
Then $\overline{\pi_\intl(\Z^nhg)}\subset\scrA$ for all $h\in H_g$, and
$\overline{\pi_\intl(\Z^nhg)}=\scrA$ for $\mu_g$-almost all $h\in H_g$.
\end{prop}
\begin{proof}
For the first claim, since $H_g$ lies in the closure of $\Gamma\varphi_g(\SL(d,\RR))$ in $G$,
it suffices to prove that if $\{\gamma_k\}\subset\Gamma$ and $\{A_k\}\subset\SL(d,\R)$ are any sequences such that
$h=\lim_k\gamma_k\varphi_g(A_k)$ exists then $\pi_\intl(\Z^nhg)\subset\overline{\pi_\intl(\Z^ng)}$.
Thus fix a vector $\vecm\in\Z^n$.
Now $\gamma_kg\varphi_1(A_k)\to hg$ and thus $\pi_\intl(\vecm\gamma_kg\varphi_1(A_k))\to\pi_\intl(\vecm hg)$
as $k\to\infty$,
and here $\pi_\intl(\vecm \gamma_kg\varphi_1(A_k))=\pi_\intl(\vecm\gamma_kg)\in\pi_\intl(\Z^ng)$ for each $k$;
hence $\pi_\intl(\vecm hg)\in\overline{\pi_\intl(\Z^ng)}$, and the claim is proved.

Replacing $\langle g,h\rangle$ by $\langle hg,h^{-1}\rangle$ in the statement just proved we conclude that 
if $h^{-1}\in H_{hg}$ then $\scrA\subset\overline{\pi_\intl(\Z^nhg)}$.
In particular if $h\in H_g$ satisfies $H_{hg}=H_g$ then $\scrA=\overline{\pi_\intl(\Z^nhg)}$.
This holds for $\mu_g$-almost all $h\in H_g$, by Corollary \ref{RATNERCOUNTABLETHMCOR}.
\end{proof}

\begin{prop}\label{ZEROSETINVPROP}
Given any affine lattice $\scrL\subset\R^n$ let us write $\scrL_0:=\scrL-\scrL$ for the lattice of which it is
a translate.
Let $g\in G$. Then
$(\Z^ng)_0\cap(\{\bn\}\times\R^m)$ is a subset of $(\Z^nhg)_0\cap(\{\bn\}\times\R^m)$ for all $h\in H_g$,
and for $\mu_g$-almost all $h\in H_g$ these two sets are equal.
\end{prop}
\begin{proof}
Assume $g=(M_g,\vecv_g)$.
Take $h=(M_h,\vecv_h)\in H_g$, and choose sequences $\{(\gamma_k,\vecm_k)\}\subset\Gamma$ and $\{A_k\}\subset\SL(d,\R)$
such that $h=\lim_k(\gamma_k,\vecm_k)\varphi_g(A_k)$;
then note that $M_h=\lim_k\gamma_k\varphi_{M_g}(A_k)$ in $G^1$.
Note also $(\Z^ng)_0=\Z^nM_g$;
consider any fixed $\vecm\in\Z^n$ such that $\vecm M_g\in\{\bn\}\times\R^m$. Then
$\vecm M_g\varphi_1(A_k^{-1})=\vecm M_g$ and therefore
$\vecm(\gamma_k\varphi_{M_g}(A_k))^{-1}=\vecm\gamma_k^{-1}\in\Z^n$.
Taking $k\to\infty$ we conclude that $\vecm M_h^{-1}=\vecm\gamma_k^{-1}\in\Z^n$ for all sufficiently large $k$, and so
$\vecm\in\Z^nM_h$ and $\vecm M_g\in\Z^n M_hM_g=(\Z^nhg)_0$.
We have thus proved that $(\Z^ng)_0\cap(\{\bn\}\times\R^m)$ is a subset of $(\Z^nhg)_0\cap(\{\bn\}\times\R^m)$.
Finally in the same way as in the proof of Proposition \ref{SCRAPROP} we conclude that we have equality 
for almost all $h\in H_g$.
\end{proof}

\begin{prop}\label{INJECTIVITYPRESERVEDPROP}
Let $g\in G$, set $\scrL=\Z^ng$ and $\scrA=\overline{\pi_\intl(\scrL)}$, and let $\scrW\neq\emptyset$ be any 
open subset of $\scrA$ such that the map
$\pi_\scrW: \{ \vecy\in\scrL\col\pi_\intl(\vecy)\in\scrW \}\to \scrP(\scrW,\scrL)$ is bijective.
Then for almost all $h\in H_g$ the corresponding map from
$\{\vecy\in\Z^nhg\col \pi_\intl(\vecy)\in\scrW\}$ to $\scrP(\scrW,\Z^nhg)$ is bijective.
\end{prop}
\begin{proof}
We first claim that for any affine lattice $\scrL'\subset\R^n$ with $\overline{\pi_\intl(\scrL')}=\scrA$, 
the restriction $\pi_{\scrW,\scrL'}$ of $\pi$ to
$\{\vecy\in\scrL'\col\pi_\intl(\vecy)\in\scrW\}$ is injective (or in other words, bijective as a map to $\scrP(\scrW,\scrL')$)
if and only if 
\begin{align}\label{INJECTIVITYPRESERVEDPROPPF1}
\scrW_0\cap\pi_\intl\bigl(\scrL_0'\cap(\{\bn\}\times\R^m)\bigr)=\{\bn\},
\end{align}
where $\scrW_0=\scrW-\scrW\subset\R^m$ and $\scrL'_0=\scrL'-\scrL'\subset\R^n$.
Indeed, if $\pi_{\scrW,\scrL'}$ is not injective then there are $\vecell_1\neq\vecell_2\in\scrL'$ satisfying
$\pi_\intl(\vecell_1),\pi_\intl(\vecell_2)\in\scrW$ and
$\pi(\vecell_1)=\pi(\vecell_2)$, and this implies
$\pi_\intl(\vecell_1)\neq\pi_\intl(\vecell_2)\in\scrW$ and 
$\pi_\intl(\vecell_1)-\pi_\intl(\vecell_2)=\pi_\intl(\vecell_1-\vecell_2)\in\pi_\intl(\scrL_0'\cap(\{\bn\}\times\R^m))$,
so that \eqref{INJECTIVITYPRESERVEDPROPPF1} fails.
Conversely, assume that \eqref{INJECTIVITYPRESERVEDPROPPF1} fails.
Then there are some $\vecw_1\neq\vecw_2\in\scrW$ and some 
$\vecell\in\scrL'_0$ such that $\pi(\vecell)=\bn$ and
$\vecw_1-\vecw_2=\pi_\intl(\vecell)$ (thus also $\vecell\neq\bn$).
Now since $\overline{\pi_\intl(\scrL')}=\scrA$, for any $\ve>0$ 
we can find some $\vecell_1\in\scrL'$ such that $\|\pi_\intl(\vecell_1)-\vecw_1\|<\ve$;
therefore, since $\scrW$ is open in $\scrA$, we can find $\vecell_1\in\scrL'$ such that
both $\pi_\intl(\vecell_1),\pi_\intl(\vecell_1-\vecell)\in\scrW$,
and now $\pi(\vecell_1)=\pi(\vecell_1-\vecell)$, i.e.\ $\pi_{\scrW,\scrL'}$ is not injective.
This completes the proof of the claim.

Using the claim and our assumptions, we have %
$\scrW_0\cap\pi_\intl\bigl(\scrL_0\cap(\{\bn\}\times\R^m)\bigr)=\{\bn\}$.
Furthermore by Propositions \ref{SCRAPROP} and \ref{ZEROSETINVPROP} we know that
$\overline{\pi_\intl(\Z^nhg)}=\scrA$ and $(\Z^nhg)_0\cap(\{\bn\}\times\R^m)=\scrL_0\cap(\{\bn\}\times\R^m)$
for almost all $h\in H_g$; and using our claim again we conclude that
$\pi_{\scrW,\Z^nhg}$ is injective for all these $h$.
\end{proof}

\section{Dynamics on the space of lattices}
\label{DYNAMICSsec}

For $\vecx\in\R^{d-1}$ and $t>0$ we write $n(\vecx)$ and $\Phi^t$ for the following elements in $\SL(d,\R)$:
\begin{align}
n(\vecx)=\matr1{\vecx}{\trans\bn}{1_{d-1}},
\qquad
\Phi^t=\matr{e^{-(d-1)t}}{\bn}{\trans\bn}{e^t1_{d-1}}.
\end{align}
For any topological space $X$ we denote by $\C_b(X)$ the space of bounded continuous functions $f:X\to\R$.

\begin{thm}\label{equi1}
Fix $g\in G$ and set $X=\Gamma\backslash\Gamma H_g$.
Let $f\in \C_b(\S_1^{d-1}\times X)$ and let $\lambda$ be a Borel probability measure on $\S_1^{d-1}$ 
which is absolutely continuous with respect to Lebesgue measure.
Then
\begin{equation}\label{equi1eq}
\lim_{t\to\infty} \int_{\S_1^{d-1}} f\Bigl(\vecv,\varphi_g(K(\vecv)\Phi^t )\Bigr) \,d\lambda(\vecv)
= \int_{\S_1^{d-1}\times X} f(\vecv,p)\,d\lambda(\vecv) \, d\mu_{H_g}(p) .
\end{equation}
\end{thm}

We will prove Theorem \ref{equi1} by extending the methods from \cite[Sec.\ 5.1-2]{partI} to the present case.
As a first step we prove the following generalization of \cite[Thm.\ 5.3]{partI}:
\begin{thm}\label{PARTITHM5P3GEN}
Fix $g\in G$ and set $X=\Gamma\backslash\Gamma H_g$.
Let $\lambda$ be a Borel probability measure on $\R^{d-1}$ which is absolutely continuous with respect to Lebesgue
measure. Let $f\in \C_b(\R^{d-1}\times X)$;
let $R$ be a subset of $\R$ having $+\infty$ as a limit point, and let $\{f_t\}_{t\in R}$
be a family of functions in $\C_b(\R^{d-1}\times X)$
which are uniformly bounded (i.e.\ $|f_t|<K$ for some absolute constant $K$) and satisfy
$f_t\to f$ as $t\to\infty$, uniformly on compacta.
Then for any $E_0\in\SL(d,\R)$ we have
\begin{align}\label{PARTITHM5P3GENRES}
\lim_{t\to\infty}\int_{\R^{d-1}} f_t\Bigl(\vecx,\varphi_g\bigl(E_0n(\vecx)\Phi^t\bigr)\Bigr)\,d\lambda(\vecx)
=\int_{\R^{d-1}\times X} f(\vecx,p)\,d\lambda(\vecx)\,d\mu_{H_g}(p).
\end{align}
\end{thm}
\begin{proof}
First assume $E_0=1_d$.
If $f(\vecx,p)\equiv F(p)$ for some $F\in \C_b(X)$ and $f_t\equiv f$ for all $t\in R$ then
\eqref{PARTITHM5P3GENRES} is a special case of Shah \cite[Thm.\ 1.4]{Shah96};
the extension to arbitrary $f,\{f_t\}_{t\in R}$ as above can be done exactly as in
\cite[Thm.\ 5.3]{partI}. 
Finally we extend to the case of general $E_0\in\SL(d,\R)$ by a simple substitution argument:
For $f,\{f_t\}_{t\in R}$ given as above, define
$\widetilde f\in \C_b(\R^{d-1}\times X)$ and 
$\{\widetilde f_t\}_{t\in R}\subset \C_b(\R^{d-1}\times X)$ through
\begin{align}
\widetilde f(\vecx,p):=f(\vecx,p\varphi_g(E_0));\qquad
\widetilde f_t(\vecx,p):=f_t(\vecx,p\varphi_g(E_0)).
\end{align}
Set $g_0=\varphi_1(E_0)$. Noticing that $\varphi_{gg_0}(A)=\varphi_{g}(E_0AE_0^{-1})$ for all $A\in\SL(d,\R)$
we see that $H_{gg_0}=H_g$.
By the limit relation which we have already proved, with 
$gg_0,\widetilde f$, $\{\widetilde f_t\}$ in the place of $g,f,\{f_t\}$, we have
\begin{align}\label{PARTITHM5P3GENPF1}
\lim_{t\to\infty}\int_{\R^{d-1}}\widetilde f_t\Bigl(\vecx,\varphi_{gg_0}\bigl(n(\vecx)\Phi^t\bigr)\Bigr)\,d\lambda(\vecx)
=\int_{\R^{d-1}\times X} \widetilde f(\vecx,p)\,d\lambda(\vecx)\,d\mu_{H_g}(p).
\end{align}
However here $\widetilde f_t(\vecx,\varphi_{gg_0}(n(\vecx)\Phi^t))=f_t(\vecx,\varphi_g(E_0n(\vecx)\Phi^t))$;
also using the fact that $\mu_{H_g}$ is right $H_g$-invariant we see that the right hand side of
\eqref{PARTITHM5P3GENPF1} equals $\int_{\R^{d-1}\times X} f(\vecx,p)\,d\lambda(\vecx)\,d\mu_{H_g}(p)$.
Hence we have proved \eqref{PARTITHM5P3GENRES}.
\end{proof}

\begin{cor}\label{PARTICOR5P4GEN}
Let $D\subset\R^{d-1}$ be an open subset and let $E_1:D\to\SO(d)$ be a smooth map such that the map 
$D\ni\vecx\mapsto\vece_1 E_1(\vecx)^{-1}\in\S_1^{d-1}$ has nonsingular differential at \mbox{(Lebesgue-)}almost all $\vecx\in D$.
Let $\lambda$ be a Borel probability measure on $D$, absolutely continuous with respect to Lebesgue measure.
Let $f\in \C_b(D\times X)$; let $R$ be a subset of $\R$ having $+\infty$ as a limit point,  and let
$\{f_t\}_{t\in R}$ be a uniformly bounded family of functions in $\C_b(D\times X)$ satisfying
$f_t\to f$ as $t\to\infty$, uniformly on compacta.
Then
\begin{align}
\lim_{t\to\infty}\int_D f_t\Bigl(\vecx,\varphi_g\bigl(E_1(\vecx)\Phi^t\bigr)\Bigr)\,d\lambda(\vecx)
=\int_{D\times X} f(\vecx,p)\,d\lambda(\vecx)\,d\mu_{H_g}(p).
\end{align}
\end{cor}
\begin{proof}
This is proved by mimicking the proof of \cite[Cor.\ 5.4]{partI}, using
Theorem \ref{PARTITHM5P3GEN} in place of \cite[Thm.\ 5.3]{partI}.
Let us only point out that \cite[eq.\ (5.23)]{partI} is now replaced by
\begin{align}
&\widetilde f_t(\vecx,p)=h(\widetilde\vecx) f_t\left(\vecx,
p\varphi_g\left(\matr{c(\vecx)^{-1}}{\bn}{\trans\vecv(\vecx)e^{-dt}}{A(\vecx)}\right)\right)
&&\text{if }\:\widetilde\vecx\in\widetilde D_0';
\\
&\widetilde f(\widetilde\vecx,p)=h(\widetilde\vecx) f\left(\vecx,
p\varphi_g\left(\matr{c(\vecx)^{-1}}{\bn}{\trans\bn}{A(\vecx)}\right)\right)
&&\text{if }\:\widetilde\vecx\in\widetilde D_0';
\\
&\widetilde f_t(\vecx,p)=\widetilde f(\vecx,p):=0
&&\text{if }\:\widetilde\vecx\notin\widetilde D_0',
\end{align}
and \cite[eq.\ (5.24)]{partI} is replaced by
\begin{align}
\lim_{t\to\infty}\int_{\R^{d-1}}\widetilde f_t\Bigl(\widetilde\vecx,\varphi_g\bigl(E_0n(\widetilde\vecx)\Phi^t\bigr)\Bigr)
\,d\widetilde\lambda(\widetilde\vecx)
=\int_{\R^{d-1}\times X}\widetilde f(\widetilde\vecx,p)\,d\widetilde\lambda(\widetilde\vecx)\,d\mu_{H_g}(p),
\end{align}
which follows from our Theorem \ref{PARTITHM5P3GEN}.
\end{proof}

\begin{proof}[Proof of Theorem \ref{equi1}]
As in \cite[beginning of Sec.\ 9.3]{partI} we may fix a smooth map $E_1:D\to\SO(d)$ such that
$\vecv=\vecv(\vecx)=\vece_1 E_1(\vecx)^{-1}$ gives a diffeomorphism between the bounded open set $D\subset\R^{d-1}$
and $\S_1^{d-1}$ minus one point, and $E_1(\vecx)=K(\vecv(\vecx))$ for all $\vecx\in D$.
Theorem \ref{equi1} is now obtained as a special case of Corollary \ref{PARTICOR5P4GEN}.
\end{proof}

We next study further the relationship between the Ratner subgroups $H_g$ and $\widetilde H_g$.

\begin{lem}\label{SCRVINTILDEHGLEM}
Let $g=(M_g,\vecv_g)\in G$, set $\scrL=\Z^nM_g$, and let $\scrA,\scrA^\circ,\scrV$ be as in the introduction.
Then $(1_n,\vecw M_g^{-1})\in\widetilde H_g$ holds for all $\vecw\in\scrV$.
\end{lem}
\begin{proof}
Note that for any $\vecm\in\Z^n$ and $\vecy\in\R^d$ we have
\begin{align}
\Gamma g\bigl(1_n,\vecm M_g+(\vecy,\bn)\bigr)g^{-1}
=\Gamma (1_n,\vecm)\varphi_g\bigl((1_d,\vecy)\bigr)
=\Gamma \varphi_g\bigl((1_d,\vecy)\bigr),
\end{align}
and this point belongs to the closed subset $\Gamma\backslash\Gamma\widetilde H_g$ of $\GaG$.
Hence also for every $\vecw$ in the closure of $\scrL+(\R^d\times\{\bn\})$
we have $g(1_n,\vecw)g^{-1}\in\Gamma\widetilde H_g$,
i.e.\ $(1_n,\vecw M_g^{-1})\in\Gamma\widetilde H_g$,
and for $\vecw$ sufficiently near $\bn$ this forces $(1_n,\vecw M_g^{-1})\in \widetilde H_g$.
Hence by linearity we have $(1_n,\vecw M_g^{-1})\in\widetilde H_g$ for all $\vecw\in\scrV$.
\end{proof}

For any $g\in G$, using the defining properties of $H_g$ and $\widetilde H_g$ and noticing that
\begin{align}\label{SLEQASLFORGENERICXPROPPF1}
\varphi_{g(1_n,(\vecx,\bn))}(A) %
=\varphi_g\bigl((A,\vecx A-\vecx)\bigr),\qquad
\forall\vecx\in\R^d,\: A\in\SL(d,\R),
\end{align}
it follows %
that $H_{g(1_n,\vecx)}\subset\widetilde H_g$ for all $\vecx\in\R^d\times\{\bn\}$.
The next proposition shows that this inclusion is in fact an \textit{equality} for almost all $\vecx$.

\begin{prop}\label{SLEQASLFORGENERICXPROP}
Let $g\in G$ be fixed. Then for (Lebesgue-)almost all $\vecx\in\R^d\times\{\bn\}$ 
we have $H_{g(1_n,\vecx)}=\widetilde H_g$.
\end{prop}
\begin{proof}
By Theorem \ref{RATNERCOUNTABLETHM}, the following family is countable:
\begin{align}
F:=\bigl\{H_{g(1_n,\vecx)}\col\vecx\in\R^d\times\{\bn\}\bigr\}.
\end{align}
As we noted above we have $H\subset\widetilde H_g$ for all $H\in F$.

Given any $H\in F$ we set
\begin{align}
V_H:=\bigl\{\vecx\in\R^d\times\{\bn\}\col H_{g(1_n,\vecx)}\subset H\bigr\}.
\end{align}
Then $\vecx\in\R^d\times\{\bn\}$ lies in $V_H$ if and only if
$\varphi_{g(1_n,\vecx)}(\SL(d,\R))\subset H$,
or in other words if and only if
$d\varphi_{g(1_n,\vecx)}(Y_j)\in\ih$ for each $j=1,\ldots,d^2-1$,
where $Y_1,\ldots,Y_{d^2-1}$ is a fixed basis of $\lsl(d,\R)$,
and $\ih$ is the Lie subalgebra of $\ig=\lasl(n,\R)$ corresponding to $H$.
Writing $g=(M_g,\vecv_g)$ we compute
\begin{align}
d\varphi_{g(1_n,\vecx)}(Y_j)&=\bigl(\Ad(M_g,\vecv_g+\vecx)\bigr)\left(\matr{Y_j}000,\bn\right)
\\\notag
&=\left(M_g\matr{Y_j}000 M_g^{-1},(\vecv_g+\vecx)\matr{Y_j}000 M_g^{-1}\right),
\end{align}
where we have identified $\ig$ in the natural way with $\lsl(n,\R)\oplus\R^n$.
It follows that for each $j$ the set of $\vecx\in\R^d\times\{\bn\}$ satisfying $d\varphi_{g(1_n,\vecx)}(Y_j)\in\ih$ 
is an affine linear subspace (i.e.\ a translate of a linear subspace) of $\R^d\times\{\bn\}$.
Hence also $V_H$ is an affine linear subspace of $\R^d\times\{\bn\}$.

Note also that if $H\in F$ satisfies $V_H=\R^d\times\{\bn\}$
then $\varphi_{g(1_n,\vecx)}(\SL(d,\R))\subset H$ for each $\vecx\in\R^d\times\{\bn\}$,
and by \eqref{SLEQASLFORGENERICXPROPPF1} this implies that $H$ contains a dense subset of
$\varphi_g(\ASL(d,\R))$; hence $\varphi_g(\ASL(d,\R))\subset H$ since $H$ is closed,
and this forces $\widetilde H_g\subset H$, i.e.\ $H=\widetilde H_g$.
We have thus proved that for each
$H\in F\setminus\{\widetilde H_g\}$, $V_H$ is an affine linear subspace of $\R^d\times\{\bn\}$,
not equal to the full set $\R^d\times\{\bn\}$.
Using the fact that $F$ is countable we conclude that $\cup_{H\in F\setminus\{\widetilde H_g\}}V_H$ has
Lebesgue measure zero in $\R^d\times\{\bn\}$.
It follows from our definitions that
for any $\vecx\in\R^d\times\{\bn\}$ outside this set we have $H_{g(1_n,\vecx)}=\widetilde H_g$.
\end{proof}

\begin{thm}\label{SLEQASLFORGENERICXPROPCOR1}
Given $g\in G$ there is a subset $\fS\subset\R^d$ of Lebesgue measure zero such that
for any $\vecq\in\R^d\setminus\fS$, any $f\in \C_b(\Gamma\backslash\Gamma \widetilde H_g)$ and any
Borel probability measure $\lambda$ on $\S_1^{d-1}$ which is absolutely continuous with respect to Lebesgue measure,
we have
\begin{equation}\label{SLEQASLFORGENERICXPROPCOR1RES}
\lim_{t\to\infty} \int_{\S_1^{d-1}} f\circ\varphi_g\big((1_d,\vecq) K(\vecv)\Phi^t\big) \,d\lambda(\vecv)
= \int_{\Gamma\backslash\Gamma \widetilde H_g} f \, d\mu_{\widetilde H_g} .
\end{equation}
\end{thm}
\begin{proof}
By Proposition \ref{SLEQASLFORGENERICXPROP} there is a set $\fS\subset\R^d$ of Lebesgue measure zero such that
$H_{g(1_n,(\vecq,\bn))}=\widetilde H_g$ holds for every $\vecq\in\R^d\setminus\fS$.
Hence by Theorem \ref{equi1}, for any $\vecq\in\R^d\setminus\fS$,
$f\in \C_b(\Gamma\backslash\Gamma \widetilde H_g)$ and any 
Borel probability measure $\lambda$ on $\S_1^{d-1}$ which is absolutely continuous with respect to Lebesgue measure,
we have
\begin{align}
\lim_{t\to\infty} \int_{\S_1^{d-1}} f\circ\varphi_{g(1_n,(\vecq,\bn))}(K(\vecv)\Phi^t ) \,d\lambda(\vecv)
= \int_{\Gamma\backslash\Gamma\widetilde H_g} f \, d\mu_{\widetilde H_g} .
\end{align}
Now the desired result follows by 
an easy substitution argument (similar to what we did in the proof of Theorem \ref{PARTITHM5P3GEN}),
using
\begin{align}
\varphi_{g(1_n,(\vecq,\bn))}(A)
=\varphi_g\bigl((1_d,\vecq)(A,\bn)\bigr)\varphi_g\bigl((1_d,-\vecq)\bigr),
\qquad\forall A\in\SL(d,\R).
\end{align}
\end{proof}

\begin{thm}\label{equi2}
Fix $g\in G$, $f\in \C_b(\Gamma\backslash\Gamma \widetilde H_g)$, a Borel probability measure $\Lambda$ on $\T^1(\RR^d)$ 
which is absolutely continuous with respect to Lebesgue measure, and $s_0>0$. Then
\begin{equation}\label{equi2RES}
\int_{\T^1(\RR^d)} f\circ\varphi_g\big((1_d,s\vecq) K(\vecv)\Phi^t\big) \,d\Lambda(\vecq,\vecv)
\to\int_{\Gamma\backslash\Gamma \widetilde H_g} f \, d\mu_{\widetilde H_g}
\qquad\text{as }\: t\to\infty,
\end{equation}
uniformly with respect to all $s\geq s_0$.
\end{thm}
\begin{proof}
For fixed $s>0$, \eqref{equi2RES} follows from Corollary \ref{SLEQASLFORGENERICXPROPCOR1} by a standard argument
using Lebesgue's Bounded Convergence Theorem, cf.\ the proof of \cite[Cor.\ 9.4]{partI}.
In order to prove uniformity with respect to $s$ we will use a compactness argument together with the fact that
$\varphi_1(\ASL(d,\R))$ commutes with all $(1_n,(\bn,\R^m))$ (cf.\ \eqref{equi2PF2} below).

Let us write $g=(M_g,\vecv_g)\in G$, set $\scrL=\Z^nM_g$, and let $\scrA,\scrA^\circ,\scrV$ be as in the introduction.
Let $C\subset\scrV$ be a closed fundamental parallelogram for $\scrV/(\scrL\cap\scrV)$.
Note that we may assume without loss of generality that $f$ has compact support,
since the extension to the more general case of bounded continuous $f$ can then be done by a standard approximation argument.
For $\vecw\in\scrV$ we define the function $f_\vecw\in\C_c(\Gamma\backslash\Gamma\widetilde H_g)$ through
\begin{align}
f_\vecw(p)=f\bigl(p(1_n,\vecw M_g^{-1})\bigr).
\end{align}
This definition is ok by Lemma \ref{SCRVINTILDEHGLEM}.
Let $\scrF$ be the closure of %
$\{f_\vecw\col\vecw\in\{\bn\}\times\pi_\intl(C)\}$ 
in $\C_b(\Gamma\backslash\Gamma\widetilde H_g)$ (with the supremum norm);
then by the Arzela-Ascoli Theorem and using the compactness of $\supp f$ and of $\pi_\intl(C)$,
we see that $\scrF$ is compact.

By the Radon-Nikodym Theorem we have $d\Lambda(\vecq,\vecv)=\lambda(\vecq,\vecv)\,d\vecq\,d\vecv$
for some $\lambda\in\L^1(\T^1(\R^d))$, and since $\C_c(\T^1(\R^d))$ is dense in $\L^1(\T^1(\R^d))$
we may assume without loss of generality that $\lambda\in \C_c(\T^1(\R^d))$.
For $\vecm\in\R^d$, $s>0$ and $\vecc\in\R^d$ we define the function
$\nu_{\vecc,\vecm,s}\in\L^1(\T^1(\R^d))$ through
\begin{align}
\nu_{\vecc,\vecm,s}(\vecq,\vecv)=\begin{cases}\lambda\bigl(s^{-1}(\vecq-\vecc+\vecm),\vecv\bigr)
&\text{if }\: \vecq-\vecc\in[0,1]^d
\\
0&\text{otherwise.}
\end{cases}
\end{align}
Let $\scrK$ be the closure of the family 
\begin{align}\label{equi2PF1}
\{\nu_{\vecc,\vecm,s}\col\vecc\in\pi(C),\:\vecm\in\R^d,\: s\geq s_0\}
\end{align}
in $\L^1(\T^1(\R^d))$.
We claim that $\scrK$ is compact.
To see this we first note that since $\lambda\in \C_c(\T^1(\R^d))$, 
the family $\scrK'=\{\nu_{\bn,\vecm,s}|_{[0,1]^d\times\S_1^{d-1}}\col\vecm\in\R^d,\: s\geq s_0\}$
is uniformly bounded and equicontinuous, and hence by the Arzela-Ascoli Theorem the closure of $\scrK'$ in
$\C([0,1]^d\times\S_1^{d-1})$ (with the supremum norm) is compact.
But every function $\mu$ in the family \eqref{equi2PF1} is given by the formula
$\mu(\vecq,\vecv)=I(\vecq-\vecc\in[0,1]^d)\nu(\vecq-\vecc,\vecv)$ for some
$\vecc\in\pi(C)$ and some $\nu\in\scrK'$,
where $I(\cdot)$ is the indicator function,
and the fact that $\scrK$ is compact follows easily from the compactness of $\pi(C)$,
the compactness of $\overline{\scrK'}$, and the fact that the $\L^1$-norm is subsumed by the supremum norm
for our compactly supported functions.

Now let $\ve>0$ be given. We have already noted that \eqref{equi2RES} holds for fixed $s$,
and applying this with $s=1$ and using the compactness of the families $\scrF$ and $\scrK$
and the fact that $\int_{\Gamma\backslash\Gamma \widetilde H_g} f_\vecw \, d\mu_{\widetilde H_g}
=\int_{\Gamma\backslash\Gamma \widetilde H_g} f \, d\mu_{\widetilde H_g}$ for each $\vecw\in\scrV$,
we conclude that there is some $T>0$ such that for all 
$t\geq T$, $\vecw\in\{\bn\}\times\pi_\intl(C)$ and $\nu\in\scrK$, we have
\begin{align}\label{equi2PF4}
\left|\int_{\T^1(\RR^d)} f_\vecw\circ\varphi_g\big((1_d,\vecq) K(\vecv)\Phi^t\big) \,\nu(\vecq,\vecv)\,
d\vecq\,d\vecv
-\int_{\T^1(\RR^d)} \nu(\vecq,\vecv)\,d\vecq\,d\vecv
\int_{\Gamma\backslash\Gamma \widetilde H_g} f \, d\mu_{\widetilde H_g}\right|<\ve.
\end{align}

Now for given $t\geq T$ and $s\geq s_0$ we note that 
\begin{align}\notag
&\int_{\T^1(\RR^d)} f\circ\varphi_g\big((1_d,s\vecq) K(\vecv)\Phi^t\big) \,d\Lambda(\vecq,\vecv)
\\\label{equi2PF3}
&=s^{-d}\sum_{\vecm\in\Z^d}\int_{[0,1]^d}\int_{\S_1^{d-1}}
f\circ\varphi_g\bigl((1_d,\vecq+\vecm)K(\vecv)\Phi^t\big)\,\lambda(s^{-1}(\vecq+\vecm),\vecv)\,d\vecv\,d\vecq.
\end{align}
For each $\vecm\in\Z^d$, since $C$ is a fundamental region for $\scrV/(\scrL\cap\scrV)$,
there is some $\veca\in\Z^n$ such that $\veca M_g\in\scrV$ and $\vecc:=(\vecm,\bn)-\veca M_g\in C$.
Let us write $\vecc=(\vecc_1,\vecc_2)\in\R^d\times\R^m=\R^n$. 
Using the fact that $(1_n,(\bn,\vecc_2))$ commutes with all $\varphi_1(\ASL(d,\R))$ we find that
\begin{align}\label{equi2PF2}
(1_n,-\veca)\,\varphi_g\bigl((1_d,\vecq+\vecm)K(\vecv)\Phi^t\bigr)
=\varphi_g\bigl((1_d,\vecq+\vecc_1)K(\vecv)\Phi^t\bigr)(1_n,(\bn,\vecc_2)M_g^{-1}).
\end{align}
Hence since $(1_n,-\veca)\in\Gamma$ we get that \eqref{equi2PF3} is equal to
\begin{multline}
s^{-d}\sum_{\vecm\in\Z^d}\int_{[0,1]^d}\int_{\S_1^{d-1}}
f_{(\bn,\vecc_2)}\circ\varphi_g\bigl((1,\vecq+\vecc_1)K(\vecv)\Phi^t\big)
\,\lambda(s^{-1}(\vecq+\vecm),\vecv)\,d\vecv\,d\vecq
\\
=s^{-d}\sum_{\vecm\in\Z^d}\int_{\T^1(\R^d)}f_{(\bn,\vecc_2)}\circ\varphi_g\bigl((1,\vecq)K(\vecv)\Phi^t\big)
\,\nu_{\vecc_1,\vecm,s}(\vecq,\vecv)\,d\vecv\,d\vecq.
\end{multline}
Here remember that $\vecc_1,\vecc_2$ depend on $\vecm$. %
By construction we have $\vecc\in C$, and thus $\vecc_1\in\pi(C)$ and $\vecc_2\in\pi_\intl(C)$, for each $\vecm\in\Z^d$.
Hence \eqref{equi2PF4} applies, and using this for each $\vecm\in\Z^d$ with $\nu_{\vecc_1,\vecm,s}\not\equiv0$ we conclude that
\begin{align}\notag
\biggl|\int_{\T^1(\RR^d)} f\circ\varphi_g\big((1,s\vecq) K(\vecv)\Phi^t\big) \,d\Lambda(\vecq,\vecv)
\hspace{150pt}
\\\label{equi2PF5}
-s^{-d}\sum_{\substack{\vecm\in\Z^d\\(\nu_{\vecc_1,\vecm,s}\not\equiv0)}}
\int_{\T^1(\R^d)}\nu_{\vecc_1,\vecm,s}(\vecq,\vecv)\,d\vecq\,d\vecv
\int_{\Gamma\backslash\Gamma \widetilde H_g} f \, d\mu_{\widetilde H_g}\biggr|
\hspace{30pt}
\\\notag
\leq s^{-d}\cdot\#\bigl\{\vecm\in\Z^d\col\nu_{\vecc_1,\vecm,s}\not\equiv0\bigr\}\cdot \ve.
\end{align}
Here we obviously have
\begin{align}
s^{-d}\sum_{\substack{\vecm\in\Z^d\\(\nu_{\vecc_1,\vecm,s}\not\equiv0)}}
\int_{\T^1(\R^d)}\nu_{\vecc_1,\vecm,s}(\vecq,\vecv)\,d\vecq\,d\vecv
=s^{-d}\int_{\T^1(\R^d)}\lambda(s^{-1}\vecq,\vecv)\,d\vecq\,d\vecv=1.
\end{align}
Furthermore we note that $\nu_{\vecc_1,\vecm,s}\not\equiv0$ can only hold when
$\vecm\in -[0,1]^d+s\cdot C_\lambda$,
where $C_\lambda\subset\R^d$ is the image of $\supp(\lambda)\subset\T^1(\R^d)$ under the %
projection $\T^1(\R^d)\to\R^d$.
Hence  \eqref{equi2PF5} implies that for all $t\geq T$,
\begin{align}
\biggl|\int_{\T^1(\RR^d)} f\circ\varphi_g\big((1,s\vecq) K(\vecv)\Phi^t\big) \,d\Lambda(\vecq,\vecv)
-\int_{\Gamma\backslash\Gamma \widetilde H_g} f \, d\mu_{\widetilde H_g}\biggr|\leq K\ve,
\end{align}
where $K$ is a constant which only depends on $\supp(\lambda)$ and $s_0$.
This concludes the proof.
\end{proof}

\section{Proof of the Siegel integral formula for quasicrystals}
\label{secSiegel}

\subsection{The Siegel integral formula}

We state the Siegel integral formula \eqref{SIF} in a slightly more general form, using affine lattices in $\RR^n$ rather than quasicrystals in $\RR^d$. Let $g\in G$ be given, and set $\scrL=\Z^ng$ and $\scrA=\overline{\pi_\intl(\scrL)}$.
For $f\in\L^1(\RR^d\times\scrA,\vol_{\R^d}\times\mu_\scrA)$,
we define the Siegel transform $\widehat f:\Gamma\backslash\Gamma H_g\to\R$ through
\begin{align}
\widehat f(\Gamma h)=\sum_{\vecm\in\ZZ^nhg\setminus\{\vecnull\}} f(\vecm).
\end{align}
(Recall that $\Z^nhg\subset\R^d\times\scrA$ for all $h\in H_g$; cf.\ Prop.\ \ref{SCRAPROP}.
It follows from the proof of the following theorem that the sum is absolutely convergent for $\mu_g$-almost every 
$\Gamma h\in\Gamma\backslash\Gamma H_g$.)

\begin{thm}\label{SIEGELTHM}
For any $f\in\L^1(\RR^d\times\scrA,\vol_{\R^d}\times\mu_\scrA)$,
\begin{equation}\label{SIF2}
\int_{\Gamma\backslash\Gamma H_g}  \widehat f(p) \,d\mu_g(p) = 
\delta_{d,m}(\scrL)\int_{\RR^d\times\scrA} f(\vecx,\vecy)\, d\!\vol_{\RR^d}(\vecx)\,d\mu_\scrA(\vecy).
\end{equation}
\end{thm}

Let us note that Theorem \ref{SIEGELRDTHM} is an immediate consequence of Theorem \ref{SIEGELTHM}.
Indeed, after a simultaneous rescaling of $\scrL$, $\scrW$ and $f$ we may assume $\delta=1$;
furthermore by taking $f(\vecx,\vecy)=I(\vecy\in\partial\scrW)$ in Theorem \ref{SIEGELTHM} and
using $\mu_\scrA(\partial\scrW)=0$ we see that we may replace $\scrW$ by $\scrW^\circ$
without affecting either side of \eqref{SIF}; thus from now on we may assume that $\scrW$ is open.
We now obtain \eqref{SIF} for a given $f_0\in\L^1(\R^d)$, by setting in \eqref{SIF2} 
\begin{equation}
f(\vecx,\vecy)=f_0(\vecx)\; I(\vecy\in\scrW),
\end{equation}
and using the fact that the restriction of $\pi$ to $\{\vecy\in\Z^nhg\col\pi_\intl(\vecy)\in\scrW\}$ is injective
for $\mu_g$-almost all $h\in H_g$ (cf.\ Proposition \ref{INJECTIVITYPRESERVEDPROP}).

\begin{proof}[Proof of Theorem \ref{SIEGELTHM}]
For $E$ running through the family of Borel sets of $\R^d\times\scrA$, the map
$E\mapsto\int_{\Gamma\backslash\Gamma H_g}\widehat\chi_E\,d\mu_g$
defines a Borel measure on $\R^d\times\scrA$, and the theorem is equivalent to the statement that %
this Borel measure equals $\delta_{d,m}(\scrL)\vol_{\R^d}\times\mu_\scrA$.
We start by considering sets of the form
$E=\scrB_r^d\times\scrW$, where $\scrW$ is any bounded open subset of $\scrA$ with $\mu_\scrA(\partial\scrW)=0$.
Note that $\widehat\chi_E$ is nonnegative and lower semicontinuous, since $E$ is open.
Hence, by Corollary \ref{PARTICOR5P4GEN} and the Portmanteau theorem 
(cf.,\ e.g., \cite[Thm.\ 1.3.4(iv)]{wellner}),
\begin{align*}
\int_{\Gamma\backslash\Gamma H_g}\widehat\chi_E(h)\,d\mu_g(h)
\leq\liminf_{R\to\infty}\int_{\SO(d)}\widehat\chi_E\bigl(\varphi_g(k\Phi^{\log R})\bigr)\,dk,
\end{align*}
where $dk$ denotes Haar measure on $\SO(d)$, normalized by $\int_{\SO(d)}dk=1$.
However, for any $R>0$ we have, by the Monotone Convergence Theorem,
\begin{align*}
\int_{\SO(d)}\widehat\chi_E\bigl(\varphi_g(k\Phi^{\log R})\bigr)\,dk
=\sum_{\vecm\in\Z^ng\setminus\{\bn\}} F_{E,R}(\vecm),
\end{align*}
where $F_{E,R}:\R^d\times\scrA\to[0,\infty]$ is given by
\begin{align*}
F_{E,R}(\vecx,\vecy)=\chi_\scrW(\vecy)\int_{\SO(d)}I\Bigl(\vecx\in \scrB_r^d\Phi^{-\log R}k^{-1}\Bigr)\,dk
=\chi_\scrW(\vecy) A_R(r^{-1}\|\vecx\|),
\end{align*}
where for $\tau>0$, $A_R(\tau)\in[0,1]$ is given by
\begin{align*}
A_R(\tau)=\frac{\vol_{\S_1^{d-1}}(\S_1^{d-1}\cap\tau^{-1}\scrB_1^d\Phi^{-\log R})}{\vol_{\S_1^{d-1}}(\S_1^{d-1})}.
\end{align*}
Let us assume $R>1$ from now on.
Note that $\scrB_1^d\Phi^{-\log R}$ is the ellipsoid $\{\vecx\col R^{-2(d-1)}x_1^2+R^2x_2^2+\ldots+R^2x_d^2<1\}$;
using this we see that $A_R(\tau)=1$ for $0<\tau\leq R^{-1}$, $A_R(\tau)=0$ for $\tau\geq R^{d-1}$,
and $A_R(\tau)$ is continuous and decreasing.
($A_R(\tau)$ may be computed explicitly in terms of an incomplete Beta function; however we do not need this.)
It follows from the above formula for $F_{E,R}(\vecx,\vecy)$ that 
\begin{align*}
\sum_{\vecm\in\Z^n g\setminus\{\bn\}} F_{E,R}(\vecm)
=\sum_{\vecm\in\scrL\setminus\{\bn\}} F_{E,R}(\vecm)
=\int_{R^{-1}}^{R^{d-1}}\#\Bigl((\scrB_{r\tau}^d\times\scrW)\cap\scrL\setminus\{\bn\}\Bigr)\,(-dA_R(\tau)),
\end{align*}
where the last integral is a Riemann-Stieltjes integral.

However because $\scrP(\scrW,\scrL)$ is uniformly discrete, there exists some $\tau_0>0$ (which depends on $\scrL$, $r$, $\scrW$) 
such that $\#((\scrB_{r\tau}^d\times\scrW)\cap\scrL\setminus\{\bn\})=0$ for all $\tau<\tau_0$.
Also, by Proposition \ref{HOFWEYLEXPLPROP}, for any given $\ve>0$ there is some $\tau_1>\tau_0$ such that for all 
$\tau\geq\tau_1$,
\begin{align*}
\#((\scrB_{r\tau}^d\times\scrW)\cap\scrL\setminus\{\bn\})\leq(1+\ve)C_\scrW r^d\tau^d,
\quad\text{where }\: C_\scrW=\delta_{d,m}(\scrL)\vol(\scrB_1^d)\mu_\scrA(\scrW).
\end{align*}
Hence for $R$ sufficiently large we have
\begin{align*}
\sum_{\vecm\in\Z^n\setminus\{\bn\}} F_{E,R}(\vecm g) &\leq O(1)\int_{\tau_0}^{\tau_1}(-dA_R(\tau))
+(1+\ve) C_\scrW r^d\int_{\tau_1}^{R^{d-1}}\tau^d(-dA_R(\tau))
\\
&=O\Bigl(A_R(\tau_0)\Bigr)+(1+\ve) C_\scrW r^d\biggl(\tau_1^d A_R(\tau_1)+
d\int_{\tau_1}^{R^{d-1}} \tau^{d-1}A_R(\tau)\,d\tau\biggr).
\end{align*}
It is clear from the definition of $A_R(\tau)$ that $A_R(\tau)\ll (R\tau)^{1-d}$ for all $\tau\geq2R^{-1}$,
and $A_R(\tau)\leq1$ for all $\tau$; hence for large $R$ the above is
\begin{align*}
=O(R^{1-d})+(1+\ve)C_\scrW dr^d\int_0^{R^{d-1}}\tau^{d-1}A_R(\tau)\,d\tau.
\end{align*}
But it is clear from the definition of $A_R(\tau)$ that 
\begin{align*}
\int_0^{R^{d-1}}\tau^{d-1}A_R(\tau)\,d\tau=\frac{\vol_{\R^d}(\scrB_1\Phi^{-\log R})}{\vol_{\S_1^{d-1}}(\S_1^{d-1})}=d^{-1}.
\end{align*}
Taking now $R\to\infty$ and then $\ve\to0$, we conclude
\begin{align}\label{SIEGELTHMPF5}
\int_{\Gamma\backslash\Gamma H_g}\widehat\chi_E(h)\,d\mu_g(h)\leq\delta_{d,m}(\scrL)(\vol_{\R^d}\times\mu_\scrA)(E),
\end{align}
for any set $E$ of the form $E=\scrB_r^d\times\scrW$.

Next, using $\varphi_g(\SL(d,\R))\subset H_g$ we see that our Borel measure 
$E\mapsto\int_{\Gamma\backslash\Gamma H_g}\widehat\chi_E\,d\mu_g$ is invariant under $\{\smatr A00{1_m}\col A\in\SL(d,\R)\}$;
also \eqref{SIEGELTHMPF5} shows that the measure is finite on any compact set.
Hence by \cite[Lemma 1.4]{Raghunathan}, for any fixed bounded Borel set $\scrW\subset\scrA$,
the Borel measure $V\mapsto\int_{\Gamma\backslash\Gamma H_g}\widehat{\chi_{V\times\scrW}}\,d\mu_g$ on $\R^d$
equals $\kappa(\scrW)\vol_{\R^d}$ for some finite constant $\kappa(\scrW)\geq0$.
Clearly our task is to prove $\kappa(\scrW)=\delta_{d,m}(\scrL)\mu_\scrA(\scrW)$,
and it suffices to prove that this holds for any bounded open subset $\scrW\subset\scrA$ with $\mu_\scrA(\partial\scrW)=0$.
Let us fix such a set $\scrW$.
By \eqref{SIEGELTHMPF5} we have $\kappa(\scrW)\leq\delta_{d,m}(\scrL)\mu_\scrA(\scrW)$.

Let $\ve>0$ be given. Let $K$ be a compact subset of $\Gamma\backslash\Gamma H_g$ with $\mu_g(K)>1-\ve$.
Since $K$ is compact, there is some $\delta>0$ such that $\|\vecm_1-\vecm_2\|\geq\delta$ for all $\Gamma h\in K$ and 
any $\vecm_1\neq\vecm_2\in\Z^nhg$.
It follows that there is a constant $C>0$ such that
\begin{align}\label{SIEGELTHMPF6}
0\leq\widehat{\chi_{\scrB_R^d\times\scrW}}(\Gamma h)\leq C(1+R)^d,\qquad\forall R>0, \: \Gamma h\in K.
\end{align}
Using also Proposition \ref{HOFWEYLEXPLPROP} and Proposition \ref{SCRAPROP} 
we conclude that for $\mu_g$-almost every $\Gamma h\in K$,
\begin{align}\label{SIEGELTHMPF7}
\frac{\widehat{\chi_{\scrB_R^d\times\scrW}}(\Gamma h)}{R^d}
\to\delta_{d,m}(\scrL)\vol(\scrB_1^d)\mu_\scrA(\scrW),
\qquad\text{as }\: R\to\infty.
\end{align}
Using \eqref{SIEGELTHMPF6} and \eqref{SIEGELTHMPF7} and the Lebesgue Dominated Convergence Theorem,
we conclude
\begin{align*}
\lim_{R\to\infty}R^{-d}\int_K\widehat{\chi_{\scrB_R^d\times\scrW}}(\Gamma h)\,d\mu_g(h)
=\mu_g(K)\delta_{d,m}(\scrL)\vol(\scrB_1^d)\mu_\scrA(\scrW).
\end{align*}
But here $\mu_g(K)>1-\ve$ and $\ve$ is arbitrarily small.
Hence we conclude that $\kappa(\scrW)\geq\delta_{d,m}(\scrL)\mu_\scrA(\scrW)$, and we are done.
\end{proof}

Let us note that Theorem \ref{SIEGELTHM} immediately implies a similar formula for $\widetilde H_g$:
\begin{cor}\label{SIEGELTILDECOR}
For any $f\in\L^1(\RR^d\times\scrA,\vol_{\R^d}\times\mu_\scrA)$,
\begin{equation}\label{SIEGELTILDECORRES}
\int_{\Gamma\backslash\Gamma \widetilde H_g} \sum_{\vecm\in\Z^nhg} f(\vecm) \,d\mu_{\widetilde H_g}(h) = 
\delta_{d,m}(\scrL)\int_{\RR^d\times\scrA} f(\vecx,\vecy)\, d\!\vol_{\RR^d}(\vecx)\,d\mu_\scrA(\vecy).
\end{equation}
\end{cor}
(Recall that $\Z^nhg\subset\R^d\times\scrA$ for all $h\in\widetilde H_g$; cf.\ Propositions \ref{SCRAPROP}
and \ref{SLEQASLFORGENERICXPROP}.)
\begin{proof}
Let $g\in G$ be given.
By Proposition \ref{SLEQASLFORGENERICXPROP} we can find $\vecz\in\R^d\times\{\bn\}$ such that $H_{g'}=\widetilde H_g$
with $g'=g(1_n,\vecz)$. Set $\scrL'=\Z^ng'=\scrL+\vecz$; then $\overline{\pi_\intl(\scrL')}=\scrA$,
since $\vecz\in\R^d\times\{\bn\}$.
Define $f_0\in\L^1(\RR^d\times\scrA,\vol_{\R^d}\times\mu_\scrA)$ through $f_0(\vecx,\vecy)=f((\vecx,\vecy)-\vecz)$.
Now by Theorem \ref{SIEGELTHM} applied to $g'$ and $f_0$ we have
\begin{align}\label{SIEGELTILDECORPF1}
\int_{\Gamma\backslash\Gamma \widetilde H_g}\sum_{\vecm\in(\Z^nhg+\vecz)\setminus\{\bn\}} f_0(\vecm) \,d\mu_{\widetilde H_g}(h) 
=\delta_{d,m}(\scrL)\int_{\RR^d\times\scrA} f_0(\vecx,\vecy)\, d\!\vol_{\RR^d}(\vecx)\,d\mu_\scrA(\vecy).
\end{align}
But using the fact that $\varphi_g((1_d,\vecx))\in\widetilde H_g$, $\forall\vecx\in\R^d$, we see that
for $\mu_{\widetilde H_g}$-almost every $h\in\widetilde H_g$ we have
$\bn\notin\Z^nhg+\vecz$. Hence the left hand side of \eqref{SIEGELTILDECORPF1} remains
unchanged if we replace $\sum_{\vecm\in(\Z^nhg+\vecz)\setminus\{\bn\}}$ by $\sum_{\vecm\in\Z^nhg+\vecz}$.
After this modification, the formula \eqref{SIEGELTILDECORPF1} is exactly the same as \eqref{SIEGELTILDECORRES}.
\end{proof}

\begin{remark}\label{XIINFTYLIMREM}
As we noted in Section \ref{SIEGELANNSEC},
the continuity for $\xi<\infty$ of the limit distributions $F_\scrP$ and $F_{\scrP,\vecq}$ in
Theorems \ref{Thm0}, \ref{Thm1} and \ref{Thm2} is an immediate consequence of 
Theorem \ref{SIEGELTHM}, Corollary \ref{SIEGELTILDECOR}, and the formulas in Theorem \ref{THM4}.
Let us now also prove the continuity at $\xi=\infty$, i.e.\ the fact that
$F_\scrP(\xi)\to0$ and $F_{\scrP,\vecq}(\xi,r)\to0$ as $\xi\to\infty$:
Since $\mu_g$ and $\mu_{\widetilde H_g}$ are $\SL(d,\R)$-invariant,
we may replace $\fZ_\xi$ by $\xi^{1/d}\fZ_1$ in \eqref{Thm3eq}, and 
replace $\fZ_\xi+r\vece_d$ by $\xi^{1/d}(\fZ_1+r\vece_d)$ in \eqref{Thm3eqA}.
As before we write $\scrP=\scrP(\scrW,\scrL)$, $\scrL=\Z^n g$, $g\in G$.
By Proposition \ref{HOFWEYLEXPLPROP} and Proposition \ref{SCRAPROP}, for $\mu_g$-almost every 
$h\in H_g$ there is some $\xi_0=\xi_0(h)>0$ such that
$\xi^{1/d}(\fZ_1+r\vece_d)\cap\scrP(\scrW,\Z^nhg)\neq\emptyset$ for all $\xi\geq\xi_0$ and all $r\in\R$.
By \eqref{Thm3eqA} this implies that $F_{\scrP,\vecq}(\xi,r)\to0$ as $\xi\to\infty$, 
uniformly with respect to $r\in\R_{\geq0}$.
The fact that $F_\scrP(\xi)\to0$ as $\xi\to\infty$ is proved in the same way, using also
the fact that $\widetilde H_g=H_{g'}$ for an appropriate $g'$, cf.\ Proposition \ref{SLEQASLFORGENERICXPROP}.

We will present detailed estimates of the tail of the limit distributions $F_\scrP$ and $F_{\scrP,\vecq}$ elsewhere;
cf.\ \cite{partIII} for the case when $\scrP$ is a lattice.
\end{remark}

\section{Proof of the limit theorems for the free path lengths}
\label{MAINPROOFSSEC}

\subsection{Proof of Theorem \ref{Thm1}}

Assume $\scrP=\scrP(\scrW,\scrL)$ and fix $g\in G^1$ and $\delta>0$ so that $\scrL=\delta^{1/n}\Z^ng$.
In fact, by an appropriate scaling of the length units, we can assume without loss of generality that $\delta=1$.

Given $(\vecq,\vecv)\in\T^1(\scrK_\rho)$ and $\xi>0$ we have
$\rho^{d-1}\tau_1(\vecq,\vecv;\rho)\geq\xi$ if and only if there is no $\scrP$-point
in the open $\rho$-neighbourhood in $\R^d$ of the line segment from
$\vecq$ to $\vecq+\rho^{1-d}\xi\vecv$. The last statement {\em implies} that $\scrP$ is disjoint from the
open cylinder $\underline{\fZ}$ of radius $\rho$ about the line segment from $\vecq$ to $\vecq+\rho^{1-d}\xi\vecv$,
and is {\em implied} whenever $\scrP$ is disjoint from the 
open cylinder $\widetilde{\fZ}$ of radius $\rho$ about the line segment from $\vecq$ to $\vecq+(\rho^{1-d}\xi+\rho)\vecv$. 
Therefore
\begin{multline}\label{mainThm1strongerPF1}
\lambda(\{ \vecv\in\S_1^{d-1}  : \widetilde{\fZ}\cap\scrP=\emptyset\})
\\ \leq \lambda(\{ \vecv\in\S_1^{d-1} : \rho^{d-1} \tau_1(\vecq,\vecv;\rho) \geq \xi \})
\\ 
\leq\lambda(\{ \vecv\in\S_1^{d-1} : \underline{\fZ}\cap\scrP=\emptyset\}) .
\end{multline}
Thus it suffices to prove that the left and right hand side of the inequality \eqref{mainThm1strongerPF1} converge to $F_\scrP(\xi)$ as $\rho\to0$ for almost every fixed $\vecq$. We will only discuss the right hand case. The left hand side can be reduced to the right hand case: we bound the left hand side from below by replacing $\widetilde{\fZ}$ by a slightly longer $\underline{\fZ}$ of length $\rho^{1-d}(\xi+\ve)$, for any $\ve>0$,
and then use $\lim_{\ve\to0}F_\scrP(\xi+\ve)=F_\scrP(\xi)$;
recall that the continuity of $F_\scrP(\xi)$ is an immediate consequence of Theorem \ref{SIEGELRDTHM}.

We have $\underline{\fZ}=\fZ_\xi \Phi^{\log\rho}K(\vecv)^{-1}(1_d,\vecq)$,
where $\fZ_\xi$ is the open cylinder of radius $1$ about the line segment from $\bn$ to $\xi\vece_1$ as defined in \eqref{zyl}.
Now %
\begin{align}\notag
&\underline{\fZ}\cap\scrP=\emptyset
\\ \label{mainThm1strongerPF1a}
&\Longleftrightarrow \Bigl(\bigl(\fZ_\xi \Phi^{\log\rho}K(\vecv)^{-1}(1_d,\vecq)\bigr)\times\scrW\Bigr)\cap\scrL=\emptyset
\\ \notag
&\Longleftrightarrow \bigl(\fZ_\xi\times\scrW\bigr)\cap\scrL \varphi_1((1_d,-\vecq)K(\vecv)\Phi^{-\log\rho})=\emptyset.
\\\notag
&\Longleftrightarrow \bigl(\fZ_\xi\times\scrW\bigr)\cap\Z^n\varphi_g((1_d,-\vecq)
K(\vecv)\Phi^{-\log\rho})g
=\emptyset.
\end{align}
Since $\scrW$ is bounded and $\mu_\scrA(\partial\scrW)=0$, $\scrW$ is Jordan measurable,
and so is the product set $\fZ_\xi\times\scrW$ as a subset of $\R^d\times\scrA$.
Hence given any $\ve>0$ there exist nonnegative continuous functions $a^-$ and $a^+$ on $\R^d\times\scrA$ 
satisfying $a^-\leq\chi_{\fZ_\xi\times\scrW}\leq a^+$
and 
\begin{align}\label{THM1pf4}
\vol_{\R^d}\times\mu_\scrA\bigl(\supp(a^+-a^-)\bigr)
<\frac{\ve}{\delta_{d,m}(\scrL)}.
\end{align}
Now define $f^+$ and $f^-\in\C_b(\Gamma\backslash\Gamma\widetilde H_g)$ through
\begin{align*}
f^\pm(\Gamma h)=\max\biggl(0,1-\sum_{\vecm\in\Z^nhg}a^\pm(\vecm)\biggr).
\end{align*}
(Again recall that $\Z^nhg\subset\R^d\times\scrA$ for all $h\in\widetilde H_g$, by
Propositions \ref{SCRAPROP} and \ref{SLEQASLFORGENERICXPROP}.)
By construction, %
\begin{align}\label{THM1pf3}
f^+(\Gamma h)\leq I\Bigl(\bigl(\fZ_\xi\times\scrW\bigr)\cap\Z^n hg=\emptyset\Bigr)\leq f^-(\Gamma h),
\qquad\forall h\in\widetilde H_g.
\end{align}
Hence by \eqref{mainThm1strongerPF1a} and Theorem \ref{SLEQASLFORGENERICXPROPCOR1},
for $\vecq$ outside a set of Lebesgue measure zero,
\begin{align}\label{THM1pf1}
\limsup_{\rho\to0}\lambda(\{ \vecv\in\S_1^{d-1} : \underline{\fZ}\cap\scrP=\emptyset\})
\leq\int_{\Gamma\backslash\Gamma\widetilde H_g} f^-(\Gamma h)\,d\mu_{\widetilde H_g}(h)
\end{align}
and
\begin{align}\label{THM1pf2}
\liminf_{\rho\to0}\lambda(\{ \vecv\in\S_1^{d-1} : \underline{\fZ}\cap\scrP=\emptyset\})
\geq\int_{\Gamma\backslash\Gamma\widetilde H_g} f^+(\Gamma h)\,d\mu_{\widetilde H_g}(h).
\end{align}
But note that we have equality throughout in \eqref{THM1pf3} for any $\Gamma h\in\Gamma\backslash\Gamma\widetilde H_g$ 
such that $a^+(\vecm)=a^-(\vecm)$ holds for all $\vecm\in\Z^nhg$. %
By \eqref{THM1pf4} and Corollary \ref{SIEGELTILDECOR}, the set of $\Gamma h\in\Gamma\backslash\Gamma\widetilde H_g$
for which this fails has measure less than $\ve$.
Note also that $f^-(\Gamma h)-f^+(\Gamma h)\leq 1$ for \textit{all} $h\in\widetilde H_g$.
Therefore the right hand sides of \eqref{THM1pf1} and \eqref{THM1pf2} are both within $\ve$ of
\begin{align*}
\mu_{\widetilde H_g}\Bigl(\Bigl\{\Gamma h\in\Gamma\backslash\Gamma\widetilde H_g\col
(\fZ_\xi\times\scrW\bigr)\cap\Z^n hg=\emptyset\Bigr\}\Bigr).
\end{align*}
Hence, since $\ve>0$ is arbitrary, and since $(\fZ_\xi\times\scrW\bigr)\cap\Z^n hg=\emptyset$ 
if and only if $\scrP(\scrW,\Z^nhg)\cap\fZ_\xi=\emptyset$, we conclude:
\begin{align*}
\lim_{\rho\to0}\lambda(\{ \vecv\in\S_1^{d-1} : \underline{\fZ}\cap\scrP=\emptyset\})
=\mu_{\widetilde H_g}\bigl(\bigl\{\Gamma h\in\Gamma\backslash\Gamma \widetilde H_g\col
\scrP(\scrW,\Z^nhg)\cap\fZ_\xi=\emptyset\bigr\}\bigr)
=F_\scrP(\xi).
\end{align*}
Cf.\ \eqref{Thm3eq} regarding the last equality.
This completes the proof of Theorem \ref{Thm1}.

\subsection{Proof of Theorem \ref{Thm0}}

The proof is virtually the same as for Theorem \ref{Thm1}, with Theorem \ref{SLEQASLFORGENERICXPROPCOR1} replaced by Theorem \ref{equi2}.

\subsection{Proof of Theorem \ref{Thm2}}

We again assume $\scrP=\scrP(\scrW,\scrL)$ and $\scrL=\Z^ng$ with $g\in G^1$.
By mimicking the argument leading to \eqref{mainThm1strongerPF1} we get:
\begin{multline}\label{mainThm2PF1}
\lambda(\{ \vecv\in\S_1^{d-1}\col \widetilde\fZ \cap\scrP=\emptyset\}) \\
\leq\lambda(\{ \vecv\in\S_1^{d-1}\col \rho^{d-1} \tau_1(\vecq+\rho\vecbeta(\vecv),\vecv;\rho) \geq \xi \})
\\
\leq\lambda(\{ \vecv\in\S_1^{d-1}\col \underline\fZ\cap\scrP=\emptyset\}),
\end{multline}
where now $\underline\fZ$ is the open cylinder of radius $\rho$ about the line segment from 
$\vecq+\rho\vecbeta(\vecv)$ to $\vecq+\rho\vecbeta(\vecv)+\rho^{1-d}\xi\vecv$
and $\widetilde\fZ$ is the open cylinder of radius $\rho$ about the line segment from 
$\vecq+\rho\vecbeta(\vecv)$ to $\vecq+\rho\vecbeta(\vecv)+(\rho^{1-d}\xi+\rho)\vecv$.
Hence, as in the proof of Theorem \ref{Thm1}, it suffices to prove that 
$\lambda(\{ \vecv\in\S_1^{d-1}\col \underline\fZ\cap\scrP=\emptyset\})$
converges to the right hand side of \eqref{Thm2eq} in Theorem \ref{Thm2},
and that the convergence is uniform with respect to the choice of $\vecq\in\scrP$.
Furthermore
we may here replace $\underline\fZ$ by the open cylinder $\fZ'$ of radius $\rho$ about the line segment from 
$\vecq+\rho\Proj_{\{\vecv\}^\perp}\vecbeta(\vecv)$ to
$\vecq+\rho\Proj_{\{\vecv\}^\perp}\vecbeta(\vecv)+\rho^{1-d}\xi\vecv$.
Using $(\Proj_{\{\vecv\}^\perp}\vecbeta(\vecv))K(\vecv)=(\vecbeta(\vecv)K(\vecv))_{\perp}$
we compute
\begin{align}
\fZ'=\fZ_{\xi,\vecv}\Phi^{\log\rho}K(\vecv)^{-1}(1_d,\vecq)
\end{align}
where
\begin{align}\label{FZXIVDEF}
\fZ_{\xi,\vecv}:=
\fZ_\xi+\bigl(\vecbeta(\vecv)K(\vecv)\bigr)_{\perp},
\end{align}
with $\fZ_\xi$ as before.
From this we get, just as in the proof of Theorem \ref{Thm1}:
\begin{align}\label{mainThm2PF3}
\fZ'\cap\scrP=\emptyset
\Longleftrightarrow
\bigl(\fZ_{\xi,\vecv}\times\scrW\bigr)\cap\Z^n\varphi_g((1_d,-\vecq)K(\vecv)\Phi^{-\log\rho})g
=\emptyset.
\end{align}
Since $\vecq\in\scrP$, there is some $\vecy\in\scrL$ such that %
$\vecq=\pi(\vecy)$.
Now $\vecy=\vecm g$ for some $\vecm\in\Z^n$,
and $(\vecq,\bn)=\vecy-(\bn,\pi_\intl(\vecy))=\vecm g-(\bn,\pi_\intl(\vecy))$.
Hence for any $h\in\ASL(d,\R)$ we have:
\begin{align}
\varphi_g((1_d,-\vecq)h)g
&=(1_n,-\vecm)\varphi_g(h)g\bigl(1_n,(\bn,\pi_\intl(\vecy))\bigr),
\end{align}
and we can rewrite \eqref{mainThm2PF3} as:
\begin{align}\label{FZPCAPPINTERP}
\fZ'\cap\scrP=\emptyset
\Longleftrightarrow
\bigl(\fZ_{\xi,\vecv}\times\scrW_\vecy\bigr)\:\cap\: 
\Z^n\varphi_g(K(\vecv)\Phi^{-\log\rho})g=\emptyset,
\end{align}
where $\scrW_\vecy:=\scrW-\pi_\intl(\vecy)$.
(Note that $\scrW_\vecy\subset\scrA$ since $\pi_\intl(\vecy)\in\scrA$.)

Now let $\ve>0$ be given, and let $a^-$ and $a^+$ be as in the proof of Theorem \ref{Thm1}.
For any $\vecv\in\S_1^{d-1}$ and $\vecz\in\scrA$ we define $a^-_{\vecv,\vecz}$ and $a^+_{\vecv,\vecz}$ to be
the appropriate translates of $a^-$ and $a^+$:
\begin{align*}
a^\pm_{\vecv,\vecz}(\vecx)=a^\pm\Bigl(\vecx+\bigl(-(\vecbeta(\vecv)K(\vecv))_\perp,\vecz\bigr)\Bigr).
\end{align*}
Note that $a^-_{\vecv,\vecz}(\vecx)$ and $a^+_{\vecv,\vecz}(\vecx)$ are jointly continuous in $\vecv,\vecz,\vecx$,
and for any $\vecv\in\S_1^{d-1}$ we have
$a^-_{\vecv,\pi_\intl(\vecy)}\leq\chi_{\fZ_{\xi,\vecv}\times\scrW_\vecy}\leq a^+_{\vecv,\pi_\intl(\vecy)}$.
We now define $f^+$ and $f^-\in\C_b(\S_1^{d-1}\times\Gamma\backslash\Gamma H_g)$ through
\begin{align*}
f^\pm(\vecv,\Gamma h)=\max\biggl(0,1-\sum_{\vecm\in\Z^nhg\setminus\{\bn\}}a_{\vecv,\pi_\intl(\vecy)}^\pm(\vecm)\biggr).
\end{align*}
Then, using the fact that $\bn\notin\fZ_{\xi,\vecv}$ for all $\vecv\in\S_1^{d-1}$,
\begin{align}\label{THM2pf3}
f^+(\vecv,\Gamma h)\leq I\Bigl(\bigl(\fZ_{\xi,\vecv}\times\scrW_\vecy\bigr)\cap\Z^nhg=\emptyset\Bigr)
\leq f^-(\vecv,\Gamma h),\qquad\forall (\vecv,h)\in \S_1^{d-1}\times H_g.
\end{align}
Hence by \eqref{FZPCAPPINTERP} and Theorem \ref{equi1} we have
\begin{align}\label{THM2pf1}
\limsup_{\rho\to0}\lambda(\{ \vecv\in\S_1^{d-1} : \fZ'\cap\scrP=\emptyset\})
\leq\int_{\S_1^{d-1}\times\Gamma\backslash\Gamma H_g} f^-(\vecv,\Gamma h)\,d\lambda(\vecv)\,d\mu_g(h)
\end{align}
and
\begin{align}\label{THM2pf2}
\liminf_{\rho\to0}\lambda(\{ \vecv\in\S_1^{d-1} : \fZ'\cap\scrP=\emptyset\})
\geq\int_{\S_1^{d-1}\times\Gamma\backslash\Gamma H_g} f^+(\vecv,\Gamma h)\,d\lambda(\vecv)\,d\mu_g(h).
\end{align}
But note that we have equality throughout in \eqref{THM2pf3} for any $(\vecv,\Gamma h)$ 
such that $a_{\vecv,\pi_\intl(\vecy)}^-(\vecm)=a_{\vecv,\pi_\intl(\vecy)}^+(\vecm)$
for all $\vecm\in\Z^n hg\setminus\{\bn\}$.
The set of $(\vecv,\Gamma h)$ for which this fails has measure bounded from above by
\begin{align*}
\int_{\S_1^{d-1}}\int_{\Gamma\backslash\Gamma H_g}\sum_{\vecm\in\Z^n hg\setminus\{\bn\}}
I\Bigl(\vecm\in\supp\bigl(a_{\vecv,\pi_\intl(\vecy)}^+-a_{\vecv,\pi_\intl(\vecy)}^-\bigr)\Bigr)\,d\mu_g(h)\,d\lambda(\vecv)
\\
=\delta_{d,m}(\scrL)\int_{\S_1^{d-1}}\vol_{\R^d}\times\mu_\scrA
\Bigl(\supp\bigl(a_{\vecv,\pi_\intl(\vecy)}^+-a_{\vecv,\pi_\intl(\vecy)}^-\bigr)\Bigr)\,d\lambda(\vecv)<\ve,
\end{align*}
where we used Theorem \ref{SIEGELTHM} and \eqref{THM1pf4} together with obvious translational invariance.
Therefore the right hand sides of \eqref{THM2pf1} and \eqref{THM2pf2} are both within $\ve$ of
\begin{align*}
\int_{\S_1^{d-1}}\mu_g\Bigl(\Bigl\{\Gamma h\in\Gamma\backslash\Gamma H_g\col
\bigl(\fZ_{\xi,\vecv}\times\scrW_\vecy\bigr)\cap\Z^nhg=\emptyset\Bigr\}\Bigr)\,d\lambda(\vecv).
\end{align*}
Hence, since $\ve>0$ is arbitrary, and since
$(\fZ_{\xi,\vecv}\times\scrW_\vecy\bigr)\cap\Z^nhg=\emptyset$ if and only if 
$\scrP(\Z^nhg,\scrW_\vecy)\cap\fZ_{\xi,\vecv}=\emptyset$, we conclude:
\begin{align}\label{THM2pf5}
\lim_{\rho\to0}\lambda(\{ \vecv\in\S_1^{d-1} : \fZ'\cap\scrP=\emptyset\})
=\int_{\S_1^{d-1}}\mu_g\Bigl(\Bigl\{\Gamma h\in\Gamma\backslash\Gamma H_g\col\scrP(\Z^nhg,\scrW_\vecy)\cap
\fZ_{\xi,\vecv}=\emptyset\Bigr\}\Bigr)\,d\lambda(\vecv).
\end{align}

Recall that we have fixed $g\in G^1$ so that $\scrL=\Z^ng$.
Now set $g'=g(1_n,(-\vecq,\bn))$, so that $\scrL-(\vecq,\bn)=\Z^ng'$
(i.e.\ $g'$ corresponds to ``$g$'' in Theorem \ref{THM4}).
Then
\begin{align*}
g'=(1_n,-\vecm)\,g\,\bigl(1_n,(\bn,\pi_\intl(\vecy))\bigr),
\end{align*}
with $\vecy,\vecm$ as above. Using %
also the fact that $(1_n,(\bn,\vecz))$ commutes with $\varphi_1(\SL(d,\R))$ for any $\vecz\in\R^m$, we now have
$\varphi_{g'}(\SL(d,\R))=\phi_\vecm(\varphi_g(\SL(d,\R)))$,   %
where $\phi_\vecm$ denotes conjugation with $(1_n,-\vecm)$,
i.e.\ $\phi_\vecm(h)=(1_n,-\vecm)h(1_n,\vecm)$ for $h\in G$.
Using $(1_n,\vecm)\in\Gamma$ it follows that $H_{g'}=\phi_\vecm(H_g)$.
Note also that for any $h\in H_g$ we have
\begin{align*}
\Z^n\phi_\vecm(h)g'=\Z^nhg+\pi_\intl(\vecy),
\end{align*}
and hence the right hand side of \eqref{THM2pf5} equals
\begin{align*}
&\int_{\S_1^{d-1}}\mu_{g'}\Bigl(\Bigl\{\Gamma h\in\Gamma\backslash\Gamma H_{g'}\col\scrP(\Z^nhg',\scrW)\cap
\fZ_{\xi,\vecv}=\emptyset\Bigr\}\Bigr)\,d\lambda(\vecv).
\end{align*}
Here the integrand is unchanged if $\fZ_{\xi,\vecv}$ is replaced by $\fZ_{\xi,\vecv}A$ for any $A\in\SL(d,\R)$
(which may depend on $\vecv$), since $\mu_{g'}$ is right $\varphi_{g'}(\SL(d,\R))$-invariant.
In particular, using this with %
an appropriate $A\in\SO(d)$, we see that $\fZ_{\xi,\vecv}$
may be replaced by $\fZ_\xi+\|(\vecbeta(\vecv)K(\vecv))_\perp\|\vece_d$, %
and therefore using \eqref{Thm3eqA} we see that the above expression equals the right hand side of \eqref{Thm2eq}.
Hence we have proved that \eqref{Thm2eq} holds for any \textit{fixed} $\vecq\in\scrP$.
Note that the proof in fact works more generally to show that \eqref{Thm2eq} holds for any point
$\vecq\in\pi(\scrL)$.   %

\vspace{5pt}

Finally we will prove that the convergence in \eqref{Thm2eq} holds uniformly over all $\vecq\in\scrP$,
and in fact more generally holds uniformly over all $\vecq\in\pi(\scrL\cap\pi_\intl^{-1}(B))$,
where $B$ is any given bounded subset of $\scrA$.
It follows from the previous discussion that
it suffices to prove that \eqref{THM2pf5} holds
uniformly over all $\vecq\in\pi(\scrL\cap\pi_\intl^{-1}(B))$.
(We now understand %
$\vecy$ to denote any point in $\scrL\cap(\{\vecq\}\times B)$;
this point is not necessarily uniquely determined by $\vecq$, but if there are more than one such $\vecy$ these all
yield the same value for the right hand side of \eqref{THM2pf5}.)

Because of the Jordan measurability of $\scrW$,
for any given $\ve>0$ we may choose the functions $a^-$ and $a^+$ on $\R^d\times\scrA$ in such a way that
\eqref{THM1pf4} holds, while the condition $a^-\leq\chi_{\fZ_\xi\times\scrW}\leq a^+$ is strengthened to
$a^-\leq\chi_{\fZ_\xi\times\scrW^-_\eta}$ and $\chi_{\fZ_\xi\times\scrW^+_\eta}\leq a^+$ for some $\eta=\eta(\ve)>0$,
where $\scrW^-_\eta=\scrW\setminus B(\partial\scrW,\eta)$ and $\scrW^+_\eta=\scrW\cup B(\partial\scrW,\eta)$,
with $B(\partial\scrW,\eta)$ denoting the $\eta$-neighbourhood of $\partial\scrW$ in $\scrA$.
Now since $B$ is bounded there is a finite set of points $\vecz_1,\ldots,\vecz_s\in\scrA$ such that
each $\vecz\in B$ lies in the $\eta$-neighborhood of some $\vecz_j$.
For each $j\in\{1,\ldots,s\}$ we define $f_j^\pm\in\C_b(\S_1^{d-1}\times\Gamma\backslash\Gamma H_g)$ by
\begin{align*}
f_j^\pm(\vecv,\Gamma h)=\max\biggl(0,1-\sum_{\vecm\in\Z^nhg\setminus\{\bn\}}a_{\vecv,\vecz_j}^\pm(\vecm)\biggr).
\end{align*}
By Theorem \ref{equi1} there is some $\rho_0>0$ such that for every $\rho\in(0,\rho_0]$ and every $j\in\{1,\ldots,s\}$,
\begin{align}\label{THM2unifpf2}
\biggl|\int_{\S_1^{d-1}}f_j^\pm\Bigl(\vecv,\varphi_g(K(\vecv)\Phi^{-\log\rho})\Bigr)\,d\lambda(\vecv)
-\int_{\S_1^{d-1}\times\Gamma\backslash\Gamma H_g}f_j^\pm(\vecv,\Gamma h)\,d\lambda(\vecv)\,d\mu_g(h)\biggr|<\ve.
\end{align}

We now claim that for every $\vecq\in\pi(\scrL\cap\pi_\intl^{-1}(B))$ and every $\rho\in(0,\rho_0]$,
\begin{align}\label{THM2unifpf1}
\left|\lambda(\{ \vecv\in\S_1^{d-1} : \fZ'\cap\scrP=\emptyset\})
-\int_{\S_1^{d-1}}\mu_g\Bigl(\Bigl\{\Gamma h\in\Gamma\backslash\Gamma H_g\col\scrP(\Z^nhg,\scrW_\vecy)\cap
\fZ_{\xi,\vecv}=\emptyset\Bigr\}\Bigr)\,d\lambda(\vecv)\right|<2\ve.
\end{align}
To prove this let $\vecq\in\pi(\scrL\cap\pi_\intl^{-1}(B))$ be given,
and fix a point $\vecy\in\scrL\cap(\{\vecq\}\times B)$.
We may now take $j\in\{1,\ldots,s\}$ such that $\|\pi_\intl(\vecy)-\vecz_j\|<\eta$,
and then by construction,
\begin{align}\label{THM2unifpf3}
f_j^+(\vecv,\Gamma h)\leq I\Bigl(\bigl(\fZ_{\xi,\vecv}\times\scrW_\vecy\bigr)\cap\Z^nhg=\emptyset\Bigr)
\leq f_j^-(\vecv,\Gamma h),\qquad\forall (\vecv,h)\in \S_1^{d-1}\times H_g.
\end{align}
Also for all $\vecv\in\S_1^{d-1}$ the equivalence \eqref{FZPCAPPINTERP} holds.
Combining these facts with \eqref{THM2unifpf2} we conclude that for all $\rho\in(0,\rho_0]$, 
\begin{align*}
\lambda(\{ \vecv\in\S_1^{d-1} : \fZ'\cap\scrP=\emptyset\})<
\int_{\S_1^{d-1}\times\Gamma\backslash\Gamma H_g}f_j^-(\vecv,\Gamma h)\,d\lambda(\vecv)\,d\mu_g(h)+\ve
\end{align*}
and
\begin{align*}
\lambda(\{ \vecv\in\S_1^{d-1} : \fZ'\cap\scrP=\emptyset\})>
\int_{\S_1^{d-1}\times\Gamma\backslash\Gamma H_g}f_j^+(\vecv,\Gamma h)\,d\lambda(\vecv)\,d\mu_g(h)-\ve.
\end{align*}
However by the same argument as before, using \eqref{THM2unifpf3},
both the last two integrals differ by at most $\ve$ from the right hand side of
\eqref{THM2pf5}; hence \eqref{THM2unifpf1} is proved.
Since $\ve>0$ was arbitrary, we have proved the desired uniformity.

\begin{appendix}
\setcounter{section}{1}
\setcounter{equation}{0}
\section*{Appendix: Directions in quasicrystals}

The methods developed in this paper can also be applied to understand the fine-scale statistics of directions in a cut-and-project set $\scrP$. In analogy with the problem of directions in affine lattices discussed in \cite[Sect.~2]{partI}, we consider the set $\scrP_T=\scrP\cap \scrB^d_T(c)\setminus\{\vecnull\}$ of points in $\scrP$ inside the spherical shell
\begin{equation}\label{shell}
	\scrB^d_T(c)=\{ \vecx\in\RR^d : c T \leq \|\vecx\| < T \} , \qquad 0\leq c < 1.
\end{equation}
In view of \eqref{density000}, there are asymptotically $C\vol(\scrB^d_1)\,T^d$ such points as $T\to\infty$, where 
\begin{equation}
C=\frac{(1-c^d)\mu_\scrA(\scrW)}{\vol(\scrV/(\scrL\cap\scrV))}
\end{equation}
and $\vol(\scrB^d_1)=\pi^{d/2}/\Gamma(\tfrac{d+2}{2})$ is the volume of the unit ball. 
For each $T$, we study the corresponding directions $\|\vecy\|^{-1}\vecy\in\S_1^{d-1}$ with $\vecy \in\scrP_T$, counted {\em with} multiplicity.
Again the asymptotics \eqref{density000} implies that, as $T\to\infty$, these points become uniformly distributed on $\S_1^{d-1}$. That is, for any set $\fU\subset\S_1^{d-1}$ with boundary of measure zero (with respect to the volume element $\vol_{\S_1^{d-1}}$ on $\S_1^{d-1}$) we have
\begin{equation}\label{udi}
\lim_{T\to\infty} \frac{\# \{\vecy\in\scrP_T\col\|\vecy\|^{-1}\vecy\in\fU\}}
{\#\scrP_T} 
= \frac{\vol_{\S_1^{d-1}}(\fU)}{\vol_{\S_1^{d-1}}(\S_1^{d-1})} . 
\end{equation}
Recall that $\vol_{\S_1^{d-1}}(\S_1^{d-1})=d \vol(\scrB^d_1)$.

To analyse the fine-scale statistics of the directions to points
in $\scrP_T$, we consider the probability of finding $r$ directions in a small 
open disc $\fD_T(\sigma,\vecv)\subset\S_1^{d-1}$ with random center $\vecv\in\S_1^{d-1}$.
Denote by
\begin{equation}\label{disc00}
	\scrN_{c,T}(\sigma,\vecv)=\#\{\vecy\in\scrP_T \col
\|\vecy\|^{-1}\vecy\in \fD_T(\sigma,\vecv)\}
\end{equation}
the number of points in $\fD_T(\sigma,\vecv)$.
The radius of $\fD_T(\sigma,\vecv)$ is chosen so that it has volume $\frac{\sigma d}{C T^{d}}$
with $\sigma>0$ fixed.
The reason for this volume scaling is that 
the expectation value for the counting function is asymptotically equal to 
$\sigma$: For any probability measure $\lambda$ on $\S_1^{d-1}$ with continuous density
\begin{equation}\label{expval}
\lim_{T\to\infty}	\int_{\S_1^{d-1}} \scrN_{c,T}(\sigma,\vecv)\, d\lambda(\vecv) =\sigma .
\end{equation}
This follows directly from \eqref{density000}.

\begin{thm}\label{visThm0}
Let $\scrP=\scrP(\scrL,\scrW)$ be a regular cut-and-project set for some (possibly affine) lattice $\scrL$. Choose $g\in G$ and $\delta>0$ so that $\scrL=\delta^{1/n}(\ZZ^n g)$. 
Let $\lambda$ be a Borel probability measure on $\S_1^{d-1}$ which is absolutely continuous with respect to Lebesgue measure. Then, for every $\sigma \geq 0$ and $r\in\ZZ_{\geq 0}$, the limit 
\begin{equation} \label{EDEF}
	E_{c,\scrP}(r,\sigma):=\lim_{T\to\infty} \lambda(\{ \vecv\in\S_1^{d-1} : \scrN_{c,T}(\sigma,\vecv)=r \})
\end{equation}
exists, and is given by
\begin{equation}
E_{c,\scrP}(r,\sigma)=
\mu_g(\{ \scrP'\in\fQ_g: \#( \scrP' \cap \fC(c,\sigma))= r \}) 
\end{equation}
where
\begin{equation} \label{FCCSDEF}
	\fC(c,\sigma) =\bigg\{(x_1,\ldots,x_d)\in\RR^d \col c < x_1 < 1, \: \|(x_2,\ldots,x_d)\|\leq 
\Bigl(\frac{\sigma d}{C\vol(\scrB_1^{d-1})}\Bigr)^{1/(d-1)}x_1\biggr\}.
\end{equation}
In particular, $E_{c,\scrP}(r,\sigma)$ is continuous in $\sigma$ and independent of $\lambda$.
\end{thm}

The proof of this theorem is analogous to that of Theorem \ref{Thm1}, with the cylinder $\fZ_\xi$ replaced by the cone $\fC(c,\sigma)$. 

Theorem \ref{visThm0} considers the set of directions in $\scrP$ {\em with} multiplicity. Although for generic $\scrP$ the multiplicity is typically one, there are important examples where this is not the case. The Penrose tiling and other cut-and-project sets which are based on the construction in Section \ref{NUMFIELDsec} fall into this category, cf.~\cite{Pleasants03}. It would therefore be natural to also consider the statistics of directions {\em without} multiplicity, in analogy with the discussion of primitive lattice points in \cite[Sect.~2.4]{partI}.
\end{appendix}

\section*{Acknowledgments} We thank Michael Baake for instructive discussions on quasicrystals, in particular for bringing Pleasants' paper \cite{Pleasants03} to our attention.
We are also grateful to Manfred Einsiedler for helpful discussions.

\end{document}